\documentclass[a4paper, 11pt]{article}
\usepackage[utf8]{inputenc}
\usepackage{fullpage}
\usepackage[margin=1in, letterpaper]{geometry}
\usepackage{my_package}
\usepackage{multirow}
\usepackage{microtype}
\usepackage{XCharter}
\usepackage{dirtytalk}
\usepackage{thm-restate}
\usepackage{enumitem}
\usepackage{fancyhdr}
\usepackage{xfrac}
\usepackage{anyfontsize}
\usepackage{todonotes}                              
\usepackage[normalem]{ulem}

\date{}

\title{Submodular Maximization over Bipartite Perfect Matchings and Matroid Intersection Bases}

 \author{Chandra Chekuri\footnote{University of Illinois Urbana-Champaign. \href{mailto:chekuri@illinois.edu}{chekuri@illinois.edu}} \and Lars Rohwedder\footnote{University of Southern Denmark. \href{mailto:rohwedder@sdu.dk}{rohwedder@sdu.dk}} \and Neta Singer\footnote{EPFL Lausanne. \href{mailto:neta.singer@epfl.ch}{neta.singer@epfl.ch}} \and Jan Vondr\'ak\footnote{Stanford University. \href{mailto:jvondrak@stanford.edu}{jvondrak@stanford.edu}} \and Rico Zenklusen\footnote{ETH Z\"urich. \href{mailto:ricoz@ethz.ch}{ricoz@ethz.ch}}}

\begin{document}

\maketitle

\begin{abstract}
Motivated by applications in fairness and foundational questions, we consider the problem of maximizing a monotone submodular function $f\colon 2^E \rightarrow \mathbb{R}_+$ over maximum cardinality sets in the intersection of two matroids on a common ground set $E$.
An important special case is submodular perfect matching in bipartite graphs.
Prior to this work, its approximability was poorly understood with 
only constant inapproximability known, despite not even a
$\sfrac{1}{o(\sqrt{|E|})}$-approximation being known. Even when allowing to violate the
cardinality constraint slightly,
only a bicriteria approximation with a significant loss in the objective was known.
Here, we obtain two results.

First, we show that, within constant factors, the problem is approximation-equivalent to Submodular Orienteering in directed graphs. This yields an $\Omega(\sfrac{1}{\log |E|})$-approximation in quasi-polynomial time together with an almost-matching hardness result.

Second, we obtain an improved polynomial-time bicriteria approximation via a local search framework.
More precisely, if $f(T^*)$ is the largest submodular value of a common independent set in both matroids of size at least $K$, we find a common independent set $T$ such that $|T| \geq (1 - \epsilon) K$ and $f(T) \geq (1/2 - \epsilon) f(T^*)$.
In contrast, previous work only guarantees a value of $\Omega(\epsilon) f(T^*)$ while ensuring that $|T| \geq (1 - \epsilon) K$.

\end{abstract}

\section{Introduction}
Submodularity plays a central role in combinatorial optimization.
Here we are concerned with submodular set function maximization subject to various combinatorial constraints. This is a well-studied topic in combinatorial optimization with many old and new applications. Classical work dates back to the work of Fisher, Nemhauser and Wolsey \cite{nemhauser1978analysis,nemhauser1978best} and there has been an explosion of work on both theoretical aspects and applications; see \cite{krause2014submodular,buchbinder2018submodular,bach2013learning,bilmes2022submodularity,dutta2026submodular} for a small sample.  
The specific problem of interest to us is when the submodular function $f\colon 2^E \to \mathbb{R}_+$ is monotone and the constraint is to find a set $S$ of cardinality \emph{at least} a given number $K$ in the intersection of two matroids $\cM_1=(E,\cI_1)$ and $\cM_2 = (E,\cI_2)$. One can, without loss of generality, assume that $K$ is the largest cardinality of a set in the intersection of the two matroids (by truncating the matroids as needed). Therefore, the problem is to maximize $f$ over common bases of $\Mcal_1$ and $\Mcal_2$.
A simple and canonical special case is when $E$ is the edge-set of a bipartite graph $G=(V,E)$ and the goal is to maximize $f$ subject to finding a \emph{perfect matching} in $G$. A key motivation for this comes from a fairness application that we will describe later.

A difficulty that arises here is that the feasible sets do not form a \emph{down-closed} family.
While known works obtain a constant-factor approximation when we seek to maximize $f$ over matchings (which form a down-closed family), only an $\Omega(1/\min\{n,\sqrt{m}\log m\})$-approximation (with $m$ being the number of edges and $n$ the number of vertices) was known for the perfect matching variant previously~\cite{mahabadi2026improved}. Additionally, the following bicriteria result was known. Given a parameter $\epsilon$, one can find a common independent set $T$ of cardinality at least $(1-\epsilon)K$ such that
$f(T) \ge \Omega(\epsilon) f(\spm)$\footnote{In fact, \cite{mahabadi2026improved} obtain an $\Omega(\epsilon)$ approximation of the optimal common independent set irrespective of its cardinality. In this work we compare to the more natural optimum over all common independent sets of size at least $K$. For this natural objective, prior work still can only guarantee $\Omega(\epsilon)$-approximation, for example when the two objectives are nearly equal.  Also, it is easy to see that if the two objectives differ greatly, one cannot achieve a good guarantee relative to the global optimum.} \cite{mahabadi2026improved}.

In this paper, we obtain two results that significantly advance our understanding of the problem.
First, we make a connection between our problem and the well-studied Submodular Orienteering problem in directed graphs \cite{ChekuriPal05}, which is another example of optimizing over a non-down-closed constraint. We show that the two problems are approximation-equivalent within a small constant factor. Via this (non-obvious) connection we obtain an $\Omega(\sfrac{1}{\log K})$-approximation in quasi-polynomial time, and also an almost matching hardness result. 
Using a recent method for submodular orienteering \cite{rohwedder2026markov}, we also obtain, for any $\epsilon > 0$, an $\Omega(1/K^\epsilon)$-approximation in time $|E|^{O(1/\epsilon)}$ in the polynomial-time regime.
Second, we obtain an improved polynomial-time bicriteria approximation via a local search framework. We show that if $\spm \in \argmax \{f(T): T \in \cI_1 \cap \cI_2, |T| \geq K \}$, the algorithm finds a solution $T$ such that $|T| \geq (1 - \epsilon) K$ and $f(T) \geq (1/2 - \epsilon) f(\spm)$.

\subsection{Related work}
Matroid constraints have received significant attention in combinatorial optimization due to their versatility and useful algorithmic properties. Maximization of a monotone submodular function subject to a single matroid constraint is well understood, with a $(1-1/e)$-approximation algorithm \cite{calinescu2011maximizing,filmus2012tight} and matching hardness result \cite{nemhauser1978analysis,feige1998threshold}. For the matroid constraint, maximizing over independent sets or bases is equivalent, since any independent set can be extended to a base (a maximum-cardinality independent set) of at least as much value due to the monotonicity property of $f$ and the extendability property in matroids. In contrast, maximization of a non-monotone submodular function over matroid bases is hard within any constant factor \cite{vondrak2013symmetry}, while the same problem over independent sets in a matroid admits a constant factor with the current best approximation ratio being $0.401$ \cite{buchbinder2024constrained}.
Many new ideas and techniques have been developed to obtain these and other
results and we refer the reader to \cite{buchbinder2018submodular} for a survey though there have been several important developments since then.

For more complicated constraints, different behaviors for maximum-cardinality and arbitrary solutions can be observed already with monotone objective functions. Consider maximization over matchings: there is a $(1/2-\epsilon)$-approximation for maximizing a monotone submodular function over matchings in a given graph \cite{lee2010submodular}, but this does not imply an algorithm for maximization over {\em perfect matchings} since not every matching can be extended to a perfect matching. This is not an issue for maximization of an additive weight function, since maximum weight perfect matchings can be reduced to maximum weight matchings, for example, by increasing the weight of each edge by the same large value.
However, this reduction is not approximation-preserving.
Although this is not an issue for additive functions, where maximum weight matching can be solved exactly, it is a significant obstacle for submodular maximization over perfect matchings.
In this paper, we address exactly this problem, and more generally the maximization of a monotone submodular function over maximum-cardinality (as well as near-maximum cardinality) sets in the intersection of two matroids.

One concrete motivation is a line of work in recent years on submodular maximization over matroid/matching constraints with additional packing, fairness, or linear constraints imposed on the solution (for example, \cite{mahabadi2026improved}, \cite{celis2017multiwinner}, \cite{halabiFairnessSubmodularMaximization2024} to name a few of the many such works). These naturally lead to constraints that are not down-closed. Past work, even in the additive setting, has often dealt with such constraints by slightly relaxing them. See for example \cite{grandoni2014new,chekuri2011multi}.

\subsection{Our contributions}
We state our results formally here. We let $K\in \mathbb{Z}_{+}$ be a cardinality parameter, which, depending on the context, we can view as an equality constraint or a lower-bound constraint.
 We now formally define the problems that we consider.

\paragraph{Submodular Matroid Intersection Basis:}
The input for this problem are two matroids $\Mcal_1 = (E, \Ical_1), \Mcal_2 = (E, \Ical_2)$ accessed via an independence oracle and a submodular function $f \colon 2^E \to \RR_+$ given by a value oracle. The goal is to find a common basis $T$ of $\Mcal_1$ and $\Mcal_2$ which maximizes $f(T)$.

We will extend this formulation to common independent sets of size $K$ for a cardinality parameter $K$. This is equivalent by truncation of the matroids. Because our statements below compare to the best common independent set of cardinality $K$ (or at least $K$), they only make sense if there is a common independent set of cardinality $K$.
We therefore assume this always to be the case.
Of course, this can be checked upfront with a classical matroid intersection algorithm.
Moreover, we denote the ground set size by $|E|$.

We will show an equivalence of this problem with Submodular Orienteering also defined here. 

\paragraph{Submodular Orienteering: } 
An instance of Submodular Orienteering consists of a directed graph $G= (V,A)$, where $V$
are the vertices and $A$ the arcs and $n \coloneq \card{V}$, and a nonnegative and monotone submodular function $f\colon 2^V \rightarrow \RR_{\geq 0}$ given by a value oracle on its vertices. Moreover, there is a start vertex $s\in V$ and an end vertex $t \in V$, which need not be distinct. Lastly, the graph is also given with arc lengths $l(u,v)$ for every $ (u,v) \in A$ and a budget $L$. The goal is to find an $s$-$t$ walk $W$ in $G$ so that the vertices $V(W)$ that are on the walk $W$ maximize $f(V(W))$ and the walk satisfies the length bound $l(W) \leq L$. 

\bigskip

\noindent
Our main bicriteria tradeoff is given by the following theorem, which relaxes the cardinality constraint.

\begin{restatable}{mythm}{bicriteriaguarantee}\label{thm:main-matroid-intersection}
For any constant $\epsilon>0$, there is a polynomial-time algorithm that, given matroids $\cM_1 = (E,\cI_1)$, $\cM_2 = (E,\cI_2)$, a monotone submodular function $f\colon 2^E \to \RR_+$, and a parameter $K \in \ZZ_+$, outputs a solution $T \in \Ical_1 \cap \Ical_2$ that satisfies $(1-\epsilon) K \leq |T| \leq K$ and $f(T) \geq \left(\frac{1}{2} - \epsilon \right) \max \{ f(S): S \in \cI_1 \cap \cI_2, |S| = K\}$.
\end{restatable}

This result is derived through a well-chosen local search procedure that starts with a common basis.
It iteratively improves the current solution, if needed by decreasing its cardinality by one unit at a time if the improvement is big enough.
To show that there exist local search improvements as long as the current solution is not good enough, we show the existence of good short swapping sets.
This is done by cutting larger improving swaps, which always exist as long as we are not optimal, into smaller pieces.
Whereas this particular step is not hard in the context of matchings, when dealing with matroid intersection, one needs to make sure that legal small swaps are obtained after cutting a legal larger swap.
We do so by leveraging a technique of \cite{lee2010submodular}.

We believe that the techniques of \Cref{thm:main-matroid-intersection} extend in a canonical way to submodular perfect (non-bipartite) matchings. This corresponds to approximating a solution of size $K = \sfrac{|V|}{2}$.
In the interest of brevity, we only focus on the matroid intersection case here, which contains all relevant ideas.
However, it does not contain submodular non-bipartite perfect matching as a special case, but only the bipartite case.

\begin{mycor}\label{coro:lowerbound-matroid-intersection}
For any constant $\epsilon>0$, there exists a polynomial-time algorithm that, given matroids $\cM_1 = (E,\cI_1)$, $\cM_2 = (E,\cI_2)$, a monotone submodular function $f\colon 2^E \to \RR_+$, and a parameter $K \in \ZZ_+$, outputs a solution $T \in \Ical_1 \cap \Ical_2$ that satisfies $|T| \geq (1-\epsilon) K$ and $f(T) \geq \left(\frac{1}{2} - \epsilon \right)  \max \{ f(S): S \in \cI_1 \cap \cI_2, |S| \geq K\}$.
\end{mycor}

\Cref{coro:lowerbound-matroid-intersection} immediately follows from applying \Cref{thm:main-matroid-intersection} for all values of $K'$ from $K$ to $|E|$ and returning the best solution. This outcome will satisfy $|T| \geq (1-\epsilon) K' \geq (1-\epsilon) K$ and the lower bound on submodular value.

The results above relax slightly the cardinality bound in order to achieve a constant factor approximation. If one wants to guarantee the cardinality $|T| = K$ exactly, then the same techniques cannot be considered. Instead, we show a connection between Submodular Matroid Intersection Basis and Submodular Orienteering.

\begin{restatable}{mythm}{orienteerreduce}\label{thm: running time guarantee exact basis}
Given matroids $\cM_1 = (E,\cI_1)$, $\cM_2 = (E,\cI_2)$, a monotone submodular function $f\colon 2^E \to \RR_+$, and a parameter $K \in \ZZ_+$, there exists an algorithm that returns a set $T$ in $\Ical_1 \cap \Ical_2$ satisfying $|T| = K$ and $f(T) \geq \frac{1}{\alpha + 1+ \epsilon} \max \{ f(S): S \in \cI_1 \cap \cI_2, |S| = K\}$ for any choice of constant $\epsilon > 0$ in polynomially many calls to an instance  of Submodular Orienteering with $n = |E|$ vertices and length function and budget consisting of polynomially bounded integers.

Here, $\sfrac 1 \alpha$ is the approximation guarantee of Submodular Orienteering. 
More precisely,
we can take $\sfrac 1 \alpha$ to be the slightly weaker guarantee that the algorithm
for Submodular Orienteering only needs to output a solution
of value at least $f(\spm_{2K}) / \alpha$, where $f(\spm_{2K}) \le f(\spm)$ is the optimum among walks that contain at most 
$2K$ edges.
\end{restatable}

Using submodular orienteering algorithms from \cite{ChekuriPal05} and \cite{rohwedder2026markov} together with \Cref{thm: running time guarantee exact basis}, we can conclude the following two corollaries.
\begin{mycor}\label{thm:quasi-algo}
 Given matroids $\cM_1 = (E,\cI_1)$, $\cM_2 = (E,\cI_2)$, a monotone submodular function $f\colon 2^E \to \RR_+$, and a parameter $K \in \ZZ_+$, there exists a quasi-polynomial time algorithm that outputs a solution $T \in \Ical_1 \cap \Ical_2$ with $|T| = K$ and $f(T) \geq \Omega\lb \apxcommonbase
    \rb \max \{ f(S): S \in \cI_1 \cap \cI_2, |S| = K\}$. 
\end{mycor}

\begin{mycor}\label{thm:orient-coro}
 Given matroids $\cM_1 = (E,\cI_1)$, $\cM_2 = (E,\cI_2)$, a monotone submodular function $f\colon 2^E \to \RR_+$, and a constant parameter $\epsilon \in \RR_{>0}$, there exists an algorithm that runs in time polynomial in $|E|^{O(1/\epsilon)}$ and outputs a solution $T \in \Ical_1 \cap \Ical_2$ with $|T| = K$ and $f(T) \geq  \Omega\lb \frac{1}{K^\epsilon}
    \rb \max \{ f(S): S \in \cI_1 \cap \cI_2, |S| = K\}$. 
\end{mycor}

Interestingly, even for the special case of monotone submodular bipartite perfect matchings, \cite{halabiFairnessSubmodularMaximization2024} exhibited a lower bound of $\Omega(\sqrt{|E|})$ on the integrality gap of the naive LP relaxation with the multilinear extension, which is a standard tool for submodular maximization.
This led \cite{mahabadi2026improved} to ask whether such a rate could be the best possible.
Our result above shows that this is not the case and significantly better approximations are achievable in quasi-polynomial time.

Because we need to guarantee the cardinality constraint exactly for \Cref{thm: running time guarantee exact basis}, we cannot rely on small swaps only to improve a solution, as with \Cref{thm:main-matroid-intersection}.
We address this challenge by using Submodular Orienteering to find strong and potentially large swaps that can be used to improve a current solution in an iterative procedure.
In the special case of perfect bipartite matchings, such swaps correspond to alternating cycles.
However, Submodular Orienteering finds a closed walk that may use elements multiple times, which leads to challenges in terms of finding a good auxiliary objective for Submodular Orienteering guaranteeing that the closed walk contains a strong swappable set that we can find efficiently.
We present an appropriate auxiliary objective, and also a way to extract a good swap from a closed walk found through Submodular Orienteering.

Complementing our positive results, we show the following hardness result via a reduction from Submodular Orienteering in the other direction. 
\begin{restatable}{mythm}{hardness}\label{thm:hardness}
	There is no quasi-polynomial time $\omega(\log\log n / \log n)$-approximation algorithm for maximizing a monotone submodular function over bipartite perfect matchings (and, in particular, over common bases of two matroids)
	assuming $\mathrm{NP}\not\subseteq \bigcap_{\epsilon > 0}\mathrm{ZPTIME}(2^{n^\epsilon})$ and the projection games conjecture holds.
\end{restatable}

$\bigcap_{\epsilon > 0}\mathrm{ZPTIME}(2^{n^\epsilon})$ is the class of problems that can be solved in expected subexponential time by a Las Vegas algorithm, that is, a zero-error randomized algorithm. For many problems in $\mathrm{NP}$, such as the satisfiability problem (SAT), only exponential time algorithms are known and $\mathrm{NP}\not\subseteq \bigcap_{\epsilon > 0}\mathrm{ZPTIME}(2^{n^\epsilon})$ is the common hypothesis that they have no subexponential time algorithms. Note that this hypothesis is weaker than $\mathrm{P}\neq\mathrm{NP}$.

The projection games conjecture is the following well-known hypothesis. It states that there exists a constant $c > 0$ such that for every $\epsilon > \sfrac{1}{n^c}$, a satisfiability instance of size $n$ can be efficiently reduced to a projection game of size $n^{1 + o(1)}\mathrm{poly}(1/\epsilon)$ with an alphabet of size $\mathrm{poly}(1/\epsilon)$ and soundness error $\epsilon$. For more details on projection games see \cite{moshkovitz2012projection}.

These complexity assumptions come from an existing lower bound for Group Steiner Tree which we are using in our hardness results.

\paragraph{Applications to fair submodular maximization:} 
We also derive several results on Fair Matroid Monotone Submodular Maximization (closely related to \cite{el2020fairness}, \cite{halabiFairnessSubmodularMaximization2024}, \cite{mahabadi2026improved}). In this framework, a submodular function is maximized over a single matroid constraint, with additional fairness quotas in the form of cardinality upper and lower bounds over subsets of the elements. Using our techniques, we are also able to derive the following theorems in this setting, where we use the shorthand $[r]\coloneqq \{1,\dots, r\}$.

\begin{restatable}{mythm}{corlowerbounds}\label{coro:matroid-constant-lowerbounds}
For any constant $\epsilon>0$, there exists a polynomial-time algorithm that, given a matroid $\cM = (E,\cI)$, a partition $E = E_1 \sqcup \cdots \sqcup E_r$ for any constant $r$, upper/lower bounds $k_i, \ell_i \in \mathbb{Z}_{\geq 0}$ for $i\in [r]$, and a monotone submodular function $f\colon 2^E \to \RR_+$, outputs a solution $T \in \Ical$ such that  $(1-\epsilon) \ell_i \leq |T \cap E_i| \leq k_i \ \forall i\in [r]$, and moreover $f(T) \geq \left(\frac{1}{2} - O(\epsilon) \right)  \max \{ f(S): S \in \cI, \ell_i \leq |S \cap E_i| \leq k_i \ \forall i \in [r]\}$ (assuming that the last-mentioned set of solutions is non-empty). 
\end{restatable}

\begin{restatable}{mythm}{exactfairness}\label{coro:exact-lower-bounds}
    There exists a quasi-polynomial time algorithm that, given a matroid $\cM = (E,\cI)$, a partition $E = E_1 \sqcup \cdots \sqcup E_r$ for any $r$, upper/lower bounds $k_i, \ell_i$ for $i\in [r]$, and a monotone submodular function $f\colon 2^E \to \RR_+$, outputs a solution $T \in \Ical$ such that  $\ell_i \leq |T \cap E_i| \leq k_i \ \forall i\in [r]$, and moreover $f(T) \geq \Omega\left(\frac{1}{\log |E|} \right)  \max \{ f(S): S \in \cI, \ell_i \leq |S \cap E_i| \leq k_i \ \forall i \in [r]\}$ (assuming that the last-mentioned set of solutions is non-empty). 
\end{restatable}

The two results are complementary: In the first one, $r$ is required to be
a constant, whereas in the second one it can be arbitrary. Furthermore,
the first result allows a violation of the lower bound, whereas the second one
does not. On the other hand, the approximation guarantee of the first
result is stronger than that of the second.

\Cref{coro:matroid-constant-lowerbounds} is proven by a modification to the bicriteria local search algorithm in \Cref{sec: corollary 2}. \Cref{thm:main-matroid-intersection} cannot be used as a black box to obtain \Cref{coro:matroid-constant-lowerbounds} because problem-specific local swaps must be considered by the algorithm to obtain the lower bounds. \Cref{coro:exact-lower-bounds} is a consequence of \Cref{thm: running time guarantee exact basis} as this problem can be reduced to finding an exact cardinality common independent set. The equivalence is shown in \Cref{sec: exact fairness}. This equivalence was observed via a different reduction in \cite{halabiFairnessSubmodularMaximization2024}, and we include ours as it provides a different and arguably simpler framework.

\section{Preliminaries}
In this work we consider maximization of a submodular function $f$ over a discrete ground set $E$. We will use $A \Delta B\coloneqq (A\setminus B)\cup (B\setminus A)$ to denote the symmetric difference of sets $A, B \subseteq E$. We denote by $[i]$ the set $\{1, 2, \hdots, i\}$ for any $i \in \mathbb{Z}_{\geq 0}$. For a set $A$ and element $i$, we use $A-i$ to denote the set $A \setminus \{i\}$ and $A + i$ to denote $A \cup \{i\}$. 

\paragraph{Submodular functions:} A submodular function $f$ over $E$ is given by $f\colon 2^{E} \rightarrow \mathbb{R}_{\geq 0}$ such that $f(A) + f(B) \geq f(A \cup B) + f(A \cap B)$ for all $A, B \subseteq E$. Moreover, $f$ is monotone if $f(A) \leq f(B)$ for all $A \subseteq B \subseteq E$. Throughout this work, we assume $f$ is normalized such that $f(\emptyset) = 0$.

\paragraph{Residual function:} Given any set $E' \subseteq E$ and submodular $f$, the residual function $f_{E'}$ is defined by $f_{E'}(A) = f(A \cup E') - f(E')$ for every $A \subseteq E$. A standard analysis guarantees that $f_{E'}$ is itself monotone submodular for any monotone submodular $f$ and $E' \subseteq E$. We will also use the notation $f(A \mid E') = f_{E'}(A)$ for any $A, E' \subseteq E$. 
\medskip

\noindent
The following properties of monotone submodular functions are well known and can be found for example in \cite{lee2010submodular}, \cite{madiman2010information}.
\begin{myfact}
    Let $A \subseteq E$ and $A = A_1 \sqcup A_2 \sqcup \hdots \sqcup A_k$ be a partition of $A$ into disjoint subsets. Then $$\displaystyle\sum_{i =1}^k f(A \setminus A_i) \geq (k-1) f(A).$$
\end{myfact}

\begin{myfact}\label{fact: covering}
    Let $\{A_1, \hdots, A_k\}$ be a collection of subsets of some $A \subseteq E$ such that every $a \in A$ is contained in exactly $q$ of the sets $A_i$. Let $B \subseteq E$ be disjoint from $A$. Then $$\displaystyle\sum_{i =1}^k (f(A_i \cup B) - f(B)) \geq q \cdot \lb f(A \cup B) - f(B) \rb.$$
\end{myfact}

\begin{myfact}\label{fact: covering2}
    Let $\{A_1, \hdots, A_k\}$ be a collection of subsets of some $A \subseteq E$ such that every $a \in A$ is contained in exactly $q$ of the sets $A_i$. Then 
    \begin{align*}
        \displaystyle\sum_{i \in k}  (f(A) - f(A \setminus A_i)) \leq q \cdot (f(A) - f(\emptyset)).
    \end{align*}
    \end{myfact}

\begin{myfact}\label{lem: frac additivity}
    Monotone submodular functions are \textit{fractionally subadditive}. In particular for any $A \subseteq E$, $f(A) \leq \sum_{i \in [k]}\beta_i f(B_i)$ for any choice of subsets $B_i \subseteq E$ and values $\beta_i \geq 0$ such that $\sum_{i \in [k]: a \in B_i} \beta_i \geq 1$ for every $a \in A$.
\end{myfact}

\paragraph{Matroids:} 
A matroid is a set family over the ground set $E$ consisting of \textit{independent sets} $\Ical \subseteq 2^E$ that satisfies 
\begin{enumerate}
    \item \textit{Downward closedness: } if $A \in \Ical$ and $B\subseteq A$ then $B \in \Ical$.
    \item \textit{Augmentation: } if $A, B \in \Ical$ and $|A| < |B|$ then there exists $e \in B \setminus A$ such that $A + e \in \Ical$.
\end{enumerate}

\paragraph{Matroid Intersection: } For two matroids $\cM_1 = (E,\cI_1), \cM_2 = (E,\cI_2)$ over a common ground set $E$, their intersection defines as feasible sets, or common independent sets, the elements in $\Ical_1 \cap \Ical_2 = \{A\subseteq E: A\in \Ical_1, A\in \Ical_2\}$. Common independent sets typically do not form the independent sets of a matroid.

A common tool for optimizing over a matroid intersection is via using the matroid intersection digraph.
For a feasible set $A \in \Ical_1 \cap \Ical_2$, we consider a bipartite directed graph $\mathcal{D}(A)$ with nodes $A \sqcup (E \setminus A)$ and directed edges as follows:
\begin{itemize}
\item For each $i \in A, j \in E \setminus A$ with $A \cup \{j\} \setminus \{i\} \in \Ical_1$, we have an arc $(i, j)$ in $\Dcal(A)$.
\item For each $i \in A, j \in E \setminus A$ with $A \cup \{j\} \setminus \{i\} \in \Ical_2$, we have an arc $(j, i)$ in $\Dcal(A)$. 
\end{itemize}
Thus left-to-right arcs encode valid swaps on $A$ with respect to $\Ical_1$ and right-to-left arcs encode valid swaps on $A$ with respect to $\Ical_2$. 
The graph contains directed paths and cycles, some of which form feasible matroid intersection exchanges. 
A dipath or dicycle $P$ is \textit{feasible} for $\Dcal(A)$ if $A \Delta V(P) \in \Ical_1 \cap \Ical_2$ and any subpath $P' \subset P$ that begins and ends at an endpoint of $P$ or a vertex in $A$ also satisfies $A \Delta V(P') \in \Ical_1 \cap \Ical_2$.

The following lemma from \cite{lee2010submodular} states that we can subdivide $\Dcal(A)$ into a collection of alternating paths and cycles that form feasible exchanges in both matroids.


\begin{mylemma}\protect{\cite[Lemma~2.5 and its proof]{lee2010submodular}}\label{lem: path cycle decomposition}
Let $A,B \in \cI_1 \cap \cI_2$ and $\Dcal(A)$ defined as above. Then there exists an $s \in \mathbb{Z}_{\geq 0}$ and a collection of feasible dipaths/dicycles $\{P_1, \hdots, P_m\}$ (possibly with repetition) in $\Dcal(A)$, using only vertices in $A \Delta B$, so that each element of $A\Delta B$ appears in exactly $2^s$ dipaths/dicycles $P_i$. Moreover, if $|B| \geq |A|$ then every dipath/dicycle $P_i$ is either a dicycle or a dipath with both endpoints in $B$. 
\end{mylemma}

Although the final property on the dipath endpoints is not explicitly stated in \cite[Lemma~2.5]{lee2010submodular}, it follows quite easily from the proof. The proof in \cite{lee2010submodular} proceeds by adding ``dummy elements'' to the smaller of the two feasible sets ($A$ in our lemma statement), and then finding a dicycle decomposition of the augmented sets. Removal of the ``dummy elements'' results in a collection of dipaths and dicycles. A dipath obtained this way cannot terminate in $A$ since that would mean we removed its neighbor in the dicycle, which would be in $B$ and therefore not a ``dummy element''.  

Using the construction in \Cref{lem: path cycle decomposition}, we may obtain the following related lemma over partition matroids specifically. 
\begin{mylemma}\label{lem: part path cycle decomposition}
    Let $A,B \in \cI_1 \cap \cI_2$ and $\Dcal(A)$ defined as above and assume $\Ical_2 = \{E' \subseteq E: |E' \cap E_j| \leq k_j \text{ for }j \in [r]\}$ is a partition matroid and that $B$ is a basis of $\Ical_2$. Then there exists a collection of dipaths/dicycles $\{P_1, \hdots, P_m\}$ satisfying the conditions of \Cref{lem: path cycle decomposition} and every right-to-left arc $(b, a)$ contained in the collection $\{P_1,\hdots, P_m\}$ corresponds to elements $a, b \in E_j$ for some part $j \in [r]$.
\end{mylemma}

\begin{proof}[Proof of \Cref{lem: part path cycle decomposition}]
    Since $|B| \geq |A|$, we extend the matroids by adding $|B| - |A|$-many ``dummy elements'' independent with everything else in the first matroid. For every part $j \in [r]$, we let $|B \cap E_j| - |A \cap E_j|$ many of the ``dummy elements'' be associated to part $j$ in the second matroid. We add all of these elements into $A$ to obtain $\tilde{A}$ which is independent in the extended matroids. Notice that after the insertion of the dummy elements, all parts are saturated by both $B$ and $\tilde{A}$, since $B$ was a basis of $\Ical_2$. Then every right-to-left arc of $\Dcal(\tilde{A})$ must have both endpoints in the same part $E_j$ for some $j \in [r]$. Now, we may use the construction of \Cref{lem: path cycle decomposition} to obtain a covering $\{P_1, \hdots, P_m\}$ of $\tilde{A} \Delta B$. Removal of the dummy elements preserves feasibility of the exchanges and as all remaining right-to-left arcs in any $P_i$ stem from right-to-left arcs in $\Dcal(\tilde{A}) \cap \Dcal(A)$, we can conclude. 
\end{proof}

Additionally, the following well-known lemma decomposes $\Dcal(A)$ into \textit{disjoint} alternating cycles, which may not necessarily be feasible. This is a direct consequence of a bijective exchange property of matroids by Brualdi \cite{brualdi1969comments}.

\begin{mylemma}\label{lem: disj cycle decomposition}
Let $A,B \in \cI_1 \cap \cI_2$ be common bases and $\Dcal(A)$ defined as above. Then there exists vertex-disjoint cycles $C_1, \hdots, C_k$ for some $k \in \mathbb{Z}_{\geq 1}$ in $\Dcal(A)$ such that 
$B = A \Delta \lb \displaystyle\bigcup_{i =1}^k V(C_i) \rb$.
\end{mylemma}

We will use the following algorithmic results which output a solution for Submodular Orienteering via a quasi-polynomial time algorithm shown in \cite{ChekuriPal05} and a polynomial time method for constant choice of $\epsilon$ shown in \cite{rohwedder2026markov}.
\begin{mythm}[\cite{ChekuriPal05}]\label{lem: orienteering-output}
    Consider a Submodular Orienteering problem over directed graph $G = (V, A,l)$ with budget $L$, start node $s$ and terminal node $t$ in $V$ (not necessarily distinct), and monotone submodular function $f$ over $2^V$.  Let $f(\spm)$ be the maximum submodular value of any $s$-$t$ walk of length at most $L$. There exists a quasi-polynomial time algorithm $\algorient$ with running time $(|V| \log L)^{O(\log |V|)} $ that outputs a walk $W$ of length at most $L$ such that $f(V(W)) \geq f(\opt) \cdot \Omega\lb \frac{1}{\log |V|}\rb$. Moreover, we can find in time $(|V|\log L)^{O(\log q)}$ a solution of value at least $\Omega\lb \frac{1}{\log q}\rb \cdot f(\spm_q)$, where $\opt_q$ is an optimal walk among solution walks with at most $q$ edges. 
\end{mythm}

\begin{mythm}[\cite{rohwedder2026markov}, Theorem 15]\label{lem: orienteering-output-new}
    Consider a Submodular Orienteering problem over directed graph $G = (V, A,l)$ with budget $L$, start node $s$ and terminal node $t$ in $V$ (not necessarily distinct), and monotone submodular function $f$ over $2^V$.  Let $f(\spm)$ be the maximum submodular value of any $s$-$t$ walk of length at most $L$. For any $\epsilon > 0$, there exists an algorithm $\algorient(\epsilon)$ with running time $|V|^{O(1/\epsilon)} $ that outputs a walk $W$ of length at most $L$ such that $f(V(W)) \geq f(\opt) \cdot \Omega\lb \frac{1}{|V|^\epsilon}\rb$. Moreover, we can find in time $|V|^{O(1/\epsilon)}$ a solution of value at least $\Omega\lb \frac{1}{q^\epsilon}\rb \cdot f(\opt_q)$, where $\opt_q$ is an optimal walk among solution walks with at most $q$ edges.
\end{mythm}

While \Cref{lem: orienteering-output-new} is not explicitly stated this way in \cite{rohwedder2026markov}, it follows by observing that the instance can be reduced to an $O(q)$-layered instance in
this case and then
Lemma 14 in \cite{rohwedder2026markov} gives the claimed guarantee with $d=q^{\Theta(\epsilon)}$ and $r =\Theta(1/\epsilon)$.

\section{Bicriteria approximation}

In this section we prove the bicriteria result Theorem~\ref{thm:main-matroid-intersection}, restated here.
\bicriteriaguarantee*

We will let
    $ \spm \in \argmax \{f(T) : T \in \Ical_1 \cap \Ical_2, \card{T} = K\}$
be an optimal solution to this problem.

We assume that this set of solutions is nonempty; otherwise our problem is vacuous.
In fact, we can also assume that $K = \max \{|T|: T \in \cI_1 \cap \cI_2 \}$. This is without loss of generality, because given $K$, we can truncate both matroids to cardinality $K$ and this yields an instance where $K$ is the maximum possible cardinality of a common independent set.

To prove the main result, Theorem~\ref{thm:main-matroid-intersection}, we use a local search method which makes two kinds of local swaps. This is formalized in the next section.

\subsection{A local search based algorithm}

We define the following algorithm, which starts with an arbitrary common independent set of $\Ical_1 \cap \Ical_2$ of size $K$, and runs local search by performing cycle swaps that are improving for $f$ or path swaps that decrease cardinality by $1$ but increase $f$ by a large additive factor. We define this as a bicriteria improving swap. Throughout the method, we assume $\epsilon$ is chosen such that $\epsilon K \in \mathbb{N}$.

\begin{mydef}
    Let $T$ be a set in $\Ical_1 \cap \Ical_2$. We say that a pair $(S', S'')$ is a \textbf{bicriteria improving swap} if $S' \subseteq T, S'' \subseteq E\setminus T$ and $(T\setminus S') \cup S'' \in \Ical_1 \cap \Ical_2$, and either
    \begin{enumerate}
        \item $|S'| \leq |S''|$ and $f((T \setminus S') \cup S'') > f(T)$ or 
        \item $|S'| = |S''| + 1$ and $f((T \setminus S') \cup S'') > f(T) + \frac{f(\spm)}{K\epsilon}$. 
    \end{enumerate}
\end{mydef}

\begin{algorithm}[H]
\DontPrintSemicolon
\SetKw{KwBy}{by}

Let $T$ be a set in $\Ical_1 \cap \Ical_2$ such that $|T| = K$.\;

\While{there exists a pair $(S', S'')$ that is a bicriteria improving swap and $|S''| \leq \frac{1}{\lambda}$} {$ T\gets (T \setminus S') \cup S''$.\;
\If{$|T| = (1-\epsilon) K$}{\Return $T$}}

\Return $T$
\caption{\textsc{Bicriteria Local Search}($(E, \Ical_1, \Ical_2), \lambda$)}
\label{alg:aux LS matching}
\end{algorithm}

\begin{remark}
   Clearly, a bicriteria improving swap can only be found if we know the value $f(\spm)$, but the value can be estimated up to a small multiplicative constant:
   We identify all elements that are not contained in any matroid intersection base. Then
   we run our method for guesses of $f(\spm) \in [M, K \cdot M]$ where $M \coloneqq \max\{f(e) : e \text{ is contained in a common independent set of size $K$} \}$ at the expense of an extra $O(\log K)$ in runtime. 
\end{remark}

\begin{remark}
    \Cref{alg:aux LS matching} can be made into a polynomial-time algorithm for any constant $ \lambda$ via standard local search modification which only makes swaps that increase the value by at least a $\delta/\mathrm{poly}(|E|)$ factor for any choice of constant $\delta>0$. The algorithm then runs in polynomial time and only loses a $(1+\delta)$ factor in the analysis (see e.g. \cite{lee2010submodular}).
\end{remark}

We now analyze the approximation guarantees of \Cref{alg:aux LS matching} with parameter $\lambda = \frac{\epsilon^2}{2}$. Our main technical lemma is the following. 

\begin{mylemma}\label{lem: local opt guarantees matroidal}
    For $\epsilon \in (0,\frac12)$, let $T$ be an output of \Cref{alg:aux LS matching} with parameter $\lambda = \frac{\epsilon^2}{2}$. Then $\card{T} \geq (1-\epsilon) K$ and $f(T) \geq (\frac{1}{2} - 2 \epsilon) f(\spm)$. 
\end{mylemma}

Note that this lemma with choice of $\epsilon' = \epsilon/2$ implies Theorem~\ref{thm:main-matroid-intersection}.

To prove the lemma, we will construct a set of local matroid intersection exchanges between $T$ and $\spm$. 
Assuming that the algorithm terminated, we can use the fact that either no local bicriteria swap is improving or that we have reached a solution of size $(1-\epsilon)K$.

\subsubsection{Proof of \Cref{lem: local opt guarantees matroidal}}
\label{sec:lemma2-proof}


Recall that $\epsilon$ has been chosen such that $\epsilon K \in \mathbb{N}$. Note that swaps never decrease the submodular value. 
If \Cref{alg:aux LS matching} terminates with $|T| = (1-\epsilon) K$, then the algorithm must have made at least $\epsilon\cdot K$ many swaps $(E', E'')$ that decrease the cardinality by $1$. Such swaps must increase the value $f((T \setminus E') \cup E'') > f(T) + \frac{f(\spm)}{K\epsilon}$, so that after $\epsilon\cdot K$ many such swaps $f(T)$ is at least $f(\spm)$ and the lemma follows. \\

The rest of the proof is dedicated to termination at $|T| > (1-\epsilon)K$ such that there do not exist any bicriteria improving swaps over $T$.

Consider $T$ found by the local search algorithm, and $\spm$ defined as above. We construct the digraph for exchanges within $\spm \Delta T$; more precisely, between $\spm \setminus T$ and $T\setminus \spm$. 
We consider the bipartite digraph $\Dcal(T)$ and a collection of dipaths/dicycles $\{Q_1, \hdots, Q_m\}$ feasible in $\Dcal(T)$ and covering every element of $T \Delta \spm$ exactly $2^s$ times for some $s \in \mathbb{Z}_{\geq 0}$, as guaranteed by \Cref{lem: path cycle decomposition}.
To construct local exchanges from this collection of dipaths/dicycles, we will need to break the paths and cycles into $2/\epsilon^2$-sized exchanges. We follow the procedure introduced in \cite{lee2010submodular}. For every feasible dipath/dicycle $Q$ in $\{Q_1, \hdots, Q_m\}$, label the vertices of $\spm \cap V(Q)$ in consecutive path order from $1$ to $|\spm \cap V(Q)|$. Then create $1/\epsilon^2$ many copies of $Q$ and, in the $i^{th}$ copy, delete all vertices whose label is congruent to $i \bmod 1/\epsilon^2$. Notice that each vertex in $Q$ of $\spm \setminus T$ is deleted in exactly 1 copy corresponding to its unique index. Altogether, this results in a collection of paths $\mathcal{C} = \{P_1, \hdots, P_t\}$ where each $P_i$ is a maximal sub-dipath/dicycle of at most $2/ \epsilon^2$ vertices of $\spm$ left intact after the deletion.

Moreover, each path $P_j$ starts and ends either at an original endpoint of $Q_i$ for some $i$ or at a vertex in $T$ (because we delete only nodes in $\spm$). Hence each path $P_j$ forms a valid matroid intersection swap, $T \Delta V(P_j) \in \Ical_1 \cap \Ical_2$. 
We will sum over all of the local swaps given by $\mathcal{C}$ to prove \Cref{lem: local opt guarantees matroidal}.

For any of the paths/cycles $P \in \mathcal{C}$, removing the vertices $V(P) \cap T$ and adding the vertices of $V(P) \cap \spm$ to $T$ maintains a feasible solution. Define $P_T = V(P) \cap T$ and $P_{O} = V(P) \cap \spm$. Since every $P\in\mathcal C$ is alternating between $T\setminus \spm$ and $\spm \setminus T$, we have $|P_T|\le |P_{O}|+1.$ Moreover, the deletion procedure guarantees that $|P_{O}|\leq 2/\epsilon^2.$ Hence the exchange $T \to T \setminus P_T \cup P_{O}$ is among the exchanges considered by the local search algorithm.

Let $\Ccal_1 \subseteq \Ccal$ be the swaps that decrease cardinality by $1$. By local optimality of $T$, for any $P \in \Ccal\setminus \Ccal_1$,
    $f(\pathswapmat) \leq f(T)$ and for any $P \in \Ccal_1$, $f(\pathswapmat) \leq f(T) + \frac{f(\spm)}{\epsilon\cdot K}$. \\

    Summing these inequalities over all $P \in \mathcal{C}$, we get the following condition
    \begin{align}
        \displaystyle\sum_{P \in \mathcal{C}}f(\pathswapmat)  \leq \displaystyle\sum_{P \in \mathcal{C}} f(T) + \displaystyle\sum_{P \in \mathcal{C}_1} \frac{f(\spm)}{\epsilon K} \label{eq: loc opt mat}.
    \end{align}

    We will use \cref{eq: loc opt mat} to prove the approximation guarantees on $T$. Note that every element in $\spm \setminus T$ appears in exactly $2^s \cdot (1/\epsilon^2 -1)$ many paths, as the corresponding node appears in $2^s$ many original paths, each of which is copied $1/\epsilon^2$ times with exactly one copy excluding the vertex. Also, every element of $T \setminus \spm$ appears in exactly $2^s \cdot 1/\epsilon^2$ many paths, as its corresponding node appears in $2^s$ original paths and in every copy. For ease of notation, let $t^\ast \coloneqq 2^s(1/\epsilon^2 -1)$ and $t \coloneqq 2^s \cdot 1/\epsilon^2$. We then rearrange and bound \cref{eq: loc opt mat} as follows:
   \begin{align}
   & \displaystyle\sum_{P \in \mathcal{C}}f(\pathswapmat)  \leq \displaystyle\sum_{P \in \mathcal{C}} f(T) + \displaystyle\sum_{P \in \mathcal{C}_1} \frac{f(\spm)}{\epsilon K} \notag\\
   \implies 
         & \displaystyle\sum_{P \in \mathcal{C}}  f(P_{O} \mid T \setminus P_T)   \leq \displaystyle\sum_{P \in \mathcal{C}}  f(P_T \mid T \setminus P_T) + \displaystyle\sum_{P \in \Ccal_1} \frac{f(\spm)}{\epsilon K}\notag \\
        \implies & \displaystyle\sum_{P \in \mathcal{C}}  f(P_{O} \mid T )   \leq \displaystyle\sum_{P \in \mathcal{C}}  f(P_T \mid T \setminus P_T) + \displaystyle\sum_{P \in \Ccal_1} \frac{f(\spm)}{\epsilon K}, \label{eq: part1 paths}
   \end{align} 
   where the second implication follows by submodularity of $f$. \\
   
  We upper bound the value $\displaystyle\sum_{P \in \mathcal{C}}  f(P_T \mid T \setminus P_T)$ and lower bound the value $\displaystyle\sum_{P \in \mathcal{C}}  f(P_{O} \mid T)$. The collection $\{P_T\}_{P \in \Ccal}$ intersects every element of $T \setminus \spm$ exactly $t$ times. Thus by \Cref{fact: covering2},
  \begin{align}\displaystyle\sum_{P \in \mathcal{C}}  f(P_T \mid T \setminus P_T)  &= \displaystyle\sum_{P \in \mathcal{C}}  (f(T) - f(T \setminus P_T)) \notag\\& \leq \displaystyle\sum_{P \in \Ccal}(f(T) - f(T \setminus P_T)) + \displaystyle\sum_{e \in T \cap \spm} t (f(T) - f(T \setminus e)) \notag\\& \leq t f(T).\label{eq: part2 paths}\end{align}
  The collection $\{P_{O}\}_{P \in {\Ccal}}$ intersects every element of $\spm \setminus T$ exactly $t^\ast$ times. Then by \Cref{fact: covering}, \begin{align}\displaystyle\sum_{P \in \mathcal{C}}  f(P_O \mid T)  = \displaystyle\sum_{P \in {\mathcal{C}}}  (f(T \cup P_O) - f(T)) \geq t^\ast (f(T \cup \spm) - f(T)).\label{eq: part3 paths}\end{align}

  Putting \cref{eq: part1 paths}, \cref{eq: part2 paths}, and \cref{eq: part3 paths} together then gives the relation
           \begin{align*}    & t^\ast(f(T \cup \spm) - f(T))   \leq t f(T) + |\Ccal_1| \frac{f(\spm)}{\epsilon K} \\
           \implies & \left( \frac{t^\ast}{t} - \frac{|\Ccal_1|}{t\epsilon K}\right)f(\spm)  \leq \left(1 + \frac{t^\ast}{t}\right)f(T).
            \end{align*}

And using the fact that $\frac{t^\ast}{t} = 1-\epsilon^2$, we get the relation
\begin{align}
   \left( 1- \epsilon^2 - \frac{|\Ccal_1|}{t\epsilon K}\right)f(\spm)  \leq \left(2 - \epsilon^2 \right)f(T).\label{eq: path-dep-bound}
\end{align}

We will just need to bound the value $\frac{|\Ccal_1|}{t\epsilon K}$ to conclude the final guarantee. Notice that $|T| \leq K = |\spm|$, so that by \Cref{lem: path cycle decomposition} all of the path exchanges in the $2^s$-sized covering $\{Q_1, \hdots, Q_m\}$ of $\Dcal(T)$ are either cycles or have both endpoints of $Q_i$ in $\spm$.

After the deletion of some vertices, we get the resulting subpaths $\Ccal$ of which an exchange in $\Ccal$ only satisfies $|P_T| = |P_O| + 1$ if both endpoints of $P$ are in $T$, since the path is alternating. Then each original neighbor in some $Q_i$ of these endpoints must have been deleted in the construction of $\Ccal$, where the original neighbors are in $\spm$. Assign to each such path $P \in \Ccal_1$ the deleted $\spm$ vertex immediately preceding the path in any fixed orientation of the paths and cycles obeying the path order. This means that every deletion of a vertex in $\spm$ is assigned uniquely to at most $1$ path in $\Ccal_1$. Notice that any element of $\spm \setminus T$ is deleted exactly once per path in $\{Q_1, \hdots, Q_m\}$ that contains the vertex, so exactly $2^s$ times. Thus $|\Ccal_1| \leq 2^s |\spm \setminus T| \leq 2^s \cdot K$. 

Then we get that 
\begin{align*}
    \frac{|\Ccal_1|}{t\epsilon K} & \leq \frac{2^s}{t \epsilon } =\epsilon.
\end{align*}

Plugging this bound into \cref{eq: path-dep-bound} shows that $f(T) \geq \lb\frac{1 - \epsilon^2 - \epsilon}{2 - \epsilon^2}\rb f(\spm) \geq \lb \frac{1}{2} - 2\epsilon \rb f(\spm)$ for sufficiently small $\epsilon$ and together with the bound on $\card{T} > (1-\epsilon) K$ this concludes the proof of the lemma.

\subsection{Modifying the local search to obtain \Cref{coro:matroid-constant-lowerbounds}}\label{sec: corollary 2}

We now modify the swaps that are considered in \Cref{alg:aux LS matching} slightly to derive \Cref{coro:matroid-constant-lowerbounds}, which, for convenience, is restated below.

\corlowerbounds*

Let $\sopt$ be an optimal solution to this problem, and we define $k'_i\coloneqq |\sopt \cap E_i |$ for $i\in [r]$. Hence, $k'_i \in \{\ell_i,\ldots,k_i\}$ for each $i\in [r]$. Assume that we can enumerate over all choices of $k'_i \in \{\ell_i,\ldots,k_i\}$.
If $r$ is constant, this can be done in $|E|^r$ many guesses. 
 For a given choice of $k'_i$, we define a partition matroid $\cI_2 \coloneqq \{ T: |T \cap E_i| \leq k'_i \ \forall i \in [r]\}$. Notice that $\sopt$ is a basis of this matroid and all bases have size $K \coloneqq \sum_{i \in [r]} k'_i$. We would then like to run the local search algorithm \Cref{alg:aux LS matching} for the intersection of the matroid constraint with this partition matroid constraint while staying close to a common basis. 
 
When multiple part lower bounds need to be maintained, we need to modify slightly our criteria for an improving swap.

\begin{mydef}
    Let $T$ be a set in $\Ical_1 \cap \Ical_2$. We say that a pair $(S', S'')$ is a \textbf{bicriteria part improving swap} if $S' \subseteq T, S'' \subseteq E\setminus T$ and $(T\setminus S') \cup S''$ is in $\Ical_1 \cap \Ical_2$, and either
    \begin{enumerate}
        \item $|S' \cap E_i| \leq |S'' \cap E_i|$ for all $i \in [r]$ and $f((T \setminus S') \cup S'') > f(T)$, or 
        \item $|S' \cap E_j| = |S'' \cap E_j| + 1$ for some $j \in [r]$, $|S' \cap E_i| \leq |S'' \cap E_i|$ for all $i\in [r]\setminus \{j\}$, and $f((T \setminus S') \cup S'') > f(T) + \frac{f(\sopt)}{k'_j \epsilon}$. 
    \end{enumerate}
\end{mydef}

Thus, an improving swap can only decrease the intersection with one part, and if it does, the submodular value must increase proportionally to this lower bound. Here we again assume $\epsilon$ is chosen such that $\epsilon k'_j \in \mathbb{N}$ for each $j \in [r]$.

We can apply \Cref{alg:aux LS matching} starting from a common basis of the two matroids with parameter $\lambda = \frac{\epsilon^2}{2}$ over bicriteria part improving swaps. We terminate the algorithm if any part intersection reaches $|T \cap E_i| = (1- \epsilon) k'_i$ and notice that this must mean that $f(T) \geq f(\sopt)$. Otherwise, we get a locally optimal output $T$ for bicriteria part improving swaps. \\ 

Applying the proof of \Cref{lem: local opt guarantees matroidal}, we get a collection of exchanges $\Ccal$ between $T$ and $\sopt$ formed by covering and deletion of some vertices of $\sopt$. We claim that these swaps are indeed feasible for bicriteria part improving swaps as well. This is formalized by the following claim.

\begin{claim}\label{claim: part-swaps}
    Let $\Ccal$ be the path exchanges between $T$ and $\sopt$ constructed in \Cref{sec:lemma2-proof}. For any $P \in \Ccal$, there exists at most one part $j \in [r]$ such that $|P_T \cap E_j| = |P_\sopt \cap E_j| + 1$ and for all other parts $i\in [r]\setminus \{j\}$ we have $|P_T \cap E_i| \leq |P_\sopt \cap E_i|$. 
\end{claim}

\begin{proof}
    
    Since $\sopt$ is a basis of $\Ical_1 \cap \Ical_2$ (after truncation), every path in the covering $\{Q_1, \hdots, Q_m\}$ is either an irreducible cycle or a path with both endpoints in $\sopt$. Moreover, by \Cref{lem: part path cycle decomposition}, all right-to-left arcs in any $Q_i$ correspond to elements of $\sopt$ and $T$ that are in the same part $E_i$ for some $i\in [r]$.
    
    Thus for every path $Q_i$ there are at least as many right-to-left arcs as elements in $V(Q_i) \cap T$ so that $|V(Q_i) \cap \sopt \cap E_j| \geq | V(Q_i) \cap T \cap E_j|$ for each part $j\in [r]$. Then for each $P \in \Ccal$, $|P_\sopt \cap E_j| < |P_T \cap E_j|$ can only happen by deletion of the neighbors in $\sopt$. Notice that every deletion of a vertex in $\sopt$ corresponds to the deletion of at most one unique right to left arc, and therefore a unique path $P \in \Ccal$ with at most one fewer right to left arcs compared to the number of nodes $|P_T|$. This corresponds to $|P_T \cap E_j| = |P_\sopt \cap E_j| + 1$ for at most one part $j \in [r]$ so the claim follows. 
\end{proof}

We now let $\Ccal_j$ correspond to the paths in $\Ccal$ such that $|P_T \cap E_j| = |P_\sopt \cap E_j| + 1$ and $j\in [r]$ is the unique such index. Following the proof of \Cref{lem: local opt guarantees matroidal} and summing over all paths in $\Ccal$, we get that 
\begin{align}
    t^\ast ((f(T \cup \sopt) - f(T)) \leq tf(T) + \displaystyle\sum_{j=1}^r |\Ccal_j| \frac{f(\sopt)}{\epsilon k'_j} \label{eq: cor 2 bound}.
\end{align}

Then, notice that any path $P \in \Ccal_j$ corresponds to deletion of a vertex in $\sopt \cap E_j$. There are exactly $2^s k'_j$ many such vertices, so $|\Ccal_j| \leq 2^s k'_j$ for each $j \in [r]$. Plugging this into \Cref{eq: cor 2 bound} gives 
\begin{align*}
    &  t^\ast ((f(T \cup \sopt) - f(T)) \leq tf(T) + \frac{2^s f(\sopt)}{\epsilon}r \\
    \implies & \lb \frac{t^\ast}{t} - \frac{2^s r}{\epsilon t}\rb f(\sopt) \leq \lb 1 + \frac{t^\ast}{t}\rb f(T) \\
    \implies &\frac{1 - \epsilon^2 - \epsilon r}{2 - \epsilon^2} f(\sopt) \leq f(T),
\end{align*}
which, if $r$ is constant and for small enough $\epsilon$, gives the desired guarantee. 

\section{Approximation for Submodular Matroid Intersection Basis}\label{sec: submod common basis}

In this section, we prove \Cref{thm: running time guarantee exact basis} and Corollaries \ref{thm:quasi-algo} and \ref{thm:orient-coro}, which state that we can find an exact cardinality common independent set that approximately maximizes a monotone submodular function by a reduction to the Submodular Orienteering problem.

As before, we assume that the set of solutions is nonempty and that $K = \max \{|T|: T \in \cI_1 \cap \cI_2 \}$ by truncation of the matroids.
We will show the approximation guarantee via a reduction of the problem to Submodular Orienteering. Recall the following relationship between the two problems stated by \Cref{thm: running time guarantee exact basis}. 

\orienteerreduce*

Note that \Cref{thm: running time guarantee exact basis} together with \Cref{lem: orienteering-output} and \Cref{lem: orienteering-output-new} directly imply \Cref{thm:quasi-algo} and \Cref{thm:orient-coro}.


The high level idea of the algorithm is to start with
a common base $B$ of $\Mcal_1$ and $\Mcal_2$ and repeatedly improve it via
an improving exchange, that is, some $R\subseteq B, A\subseteq E\setminus B$ such that $(B\setminus R)\cup A$ is again a common base and $f((B\setminus R)\cup A) > f(B)$, until
a good approximation is reached. 
An exchange forms a set of simple cycles in $\Dcal(B)$, see \Cref{lem: disj cycle decomposition}, which each correspond to alternatingly adding and removing elements. Since it is improving, we gain more value from adding elements than we lose from removing elements.
Intuitively, we want to construct an instance of Submodular Orienteering on the graph $\Dcal(B)$ that
models the task of finding such an improving exchange.
To model the gain from adding elements, we use $f_B$ as the objective function in
the orienteering instance; to model the loss from removing elements, we 
use a carefully chosen length function and a budget that we will specify later. The key idea is that the loss is modeled as an additive function which helps simplify the algorithm and analysis. 
We search via the orienteering problem for a walk starting and ending
at one of the nodes $b \in B$.
We think of $b$ as an arbitrary vertex in a cycle corresponding to a suitable exchange.
We defer details on how to obtain $b$ to later.

A walk given by the solution to the Submodular Orienteering instance may visit vertices multiple times. Thus, we may have to  decompose it into simple cycles. In bipartite perfect matching, all simple cycles form feasible exchanges, but in the more general
setting of common bases of two matroids this may not hold. Thus, we need to further translate the cycles into feasible exchanges. Our main technical lemma states that if we find a $b$-$b$ walk $W$ for some $b \in B$ with large submodular value compared to the length, then we can derive a feasible exchange that is improving for $B$. 

\subsection{From cycles to feasible exchanges}
%
Recall that cycles need not be feasible matroid exchanges.
To address this challenge, we use a decomposition idea that splits a cycle in $\Dcal(B)$ into a collection of feasible exchanges contained in the cycle. If the cycle was improving for $B$, then one of these exchanges will also be improving for $B$ and feasible. The decomposition idea is similar to related
techniques in~\cite{chekuri2009dependent, chekuri2011multi}. However, these prior works decompose a (collection of) cycles into irreducible cycles which form a fractional covering by feasible exchanges. Here we use an (arguably) simpler decomposition into exchanges which need not be irreducible cycles and which are therefore more directly convenient for our purposes and perhaps in future uses. 
It is formalized by the following lemma.

\begin{mylemma}\label{lem: fractional cycle repair}
    Let $C$ be a simple directed cycle in $\Dcal(B)$. Let $C_B = V(C) \cap B$ and $\cnotb = V(C) \cap (E \setminus B)$. Then there exists common bases $\{B_1, \hdots, B_m\}$ such that for $j \in [m]$, 
    $$B_j = (B \setminus R_j) \cup A_j \text{ for some } R_j \subseteq C_B, A_j \subseteq \cnotb.$$ Moreover, there are coefficients $\lambda_j \geq 0$ for each $j \in [m]$ such that
    $$\displaystyle\sum_{\substack{j\in [m]:\\ b \in R_j}}\lambda_j = 1 \ \forall b \in C_B, \quad \displaystyle\sum_{\substack{j\in [m]:\\ e \in A_j}} \lambda_j = 1 \ \forall e \in \cnotb, \quad \text{and } \displaystyle\sum_{j=1}^m \lambda_j = |C_B| = |\cnotb|,$$
    and these bases and coefficients can be found in polynomial time in $n$. 
\end{mylemma}

\begin{proof}[Proof of \Cref{lem: fractional cycle repair}]
    Given $C$, a simple directed cycle in $\Dcal(B)$, let $C_B = \{b_1, \hdots, b_k\}$ and $\cnotb = \{e_1, \hdots, e_k\}$ such that $b_i \rightarrow e_i$ is a left-to-right arc in $C$ for each $i$.
    Then $B_i \coloneqq (B \setminus \{b_i\}) \cup \{e_i\}$ is a base in $\Mcal_1$ for each $i \in [k]$. The point $\frac{1}{k}\sum_{i =1}^k \mathbbm{1}_{B_i}$ is therefore in the base polytope of $\Mcal_1$.
    Note that this point can be written as
    \begin{equation*}
    \frac{1}{k}\displaystyle\sum_{i =1}^k \mathbbm{1}_{B_i} = \mathbbm{1}_B - \frac{1}{k}\displaystyle\sum_{b \in C_B} \mathbbm{1}_{\{b\}} + \frac{1}{k}\displaystyle\sum_{e \in \cnotb} \mathbbm{1}_{\{e\}}.
    \end{equation*} 
    An identical argument over right-to-left arcs and $k$ bases in $\Mcal_2$ shows that this same point is in the base polytope of $\Mcal_2$. By integrality of the intersection of the 2 matroid base polytopes, this point can be written as a convex combination of \textit{common bases} of $\Mcal_1$ and $\Mcal_2$. In particular, 
    \begin{equation}\label{eq:decomp_into_Bjs}
     \mathbbm{1}_B - \frac{1}{k}\displaystyle\sum_{b \in C_B} \mathbbm{1}_{\{b\}} + \frac{1}{k}\displaystyle\sum_{e \in \cnotb} \mathbbm{1}_{\{e\}} = \displaystyle\sum_{j=1}^m \mu_j \mathbbm{1}_{B_j}
     \end{equation}
    for common bases $B_1, \hdots, B_m$ for some $m \geq 1$ and $\mu_j \geq 0$ such that $\sum_{j=1}^m \mu_j = 1$. The left-hand side of~\eqref{eq:decomp_into_Bjs} forces every $B_j$ to be of the form $(B \setminus R_j) \cup A_j$ for some $R_j \subseteq C_B$ and $A_j \subseteq \cnotb$.
    Moreover, restricting to any coordinate $b \in C_B$ shows that
   \begin{equation*} 
    1 - \frac{1}{k} = \sum_{\substack{j\in [m]:\\ b \notin R_j}} \mu_j,
    \end{equation*}
    and therefore
    \begin{equation*}
    \sum_{\substack{j\in [m]:\\ b \in R_j}} \mu_j = \frac{1}{k} \qquad \forall b \in C_B.
    \end{equation*}
    Analogously, restricting to any coordinate $e \in \cnotb$ shows that
    \begin{equation*}
     \sum_{\substack{j\in [m]:\\ e \in A_j}} \mu_j = \frac{1}{k}.
    \end{equation*}
    Setting $\lambda_j = k \mu_j$ for every $j \in [m]$ concludes the lemma. Polynomial time computation follows by decomposing the given point into a convex combination of bases in the matroid intersection base polytope using weighted matroid intersection algorithms (see, e.g., \cite{schrijver2003combinatorial}). 
\end{proof}

We now show how to find improving cycles in $\Dcal(B)$, which we can then decompose into feasible exchanges using \Cref{lem: fractional cycle repair}. 

\subsection{Finding an improving exchange via Submodular Orienteering}
We need to provide the missing details on the construction of the
instance of Submodular Orienteering and how we translate a solution into an
improving exchange.

First, we specify an initial arc length function $l$ for arcs in $\Dcal(B)$ to model
the loss from removing elements. (We will actually use a slight variation of it later.)
Label the elements of $B$ by an arbitrary order such that $B = \{1, \hdots, |B|\}$. For every $i \in B$, and every arc $(u, i)$, we set $l(u,i) \coloneqq f([i]) - f([i-1])$.
For every $e \in E\setminus B$ and any arc $(u, e)$, set the arc length of $(u, e)$ to $0$. Note that by monotonicity of $f$, all vertex lengths are nonnegative. And arcs directing into vertices of $B$ are the only arcs with positive length. We claim that these lengths are an overestimate for the submodular value lost by removing any subset of $B$. Since the arc length is determined by its endpoint, we define for any vertex set $S \subseteq V$ the length of traversing a cycle through this vertex set by
\begin{equation*}
l(S) \coloneqq \sum_{i \in S\cap B} \left(f([i]) - f([i-1])\right).
\end{equation*}
\begin{claim}\label{claim: modular lower support}
    $f(R \mid (B\setminus R)) \leq l(R) \quad \forall R\subseteq B$.
\end{claim}
\begin{proof}
    We have that for any $i \in B \setminus R$, \begin{align*}
        f(\{i\} \mid ([i-1]\cap (B\setminus R))) & \geq f(\{i\} \mid [i-1]) \ge l(\{i\}) .
    \end{align*} 
    Thus, $$l(B\setminus R) = \displaystyle\sum_{i \in B\setminus R}l(\{i\}) \leq \displaystyle\sum_{i \in B\setminus R} f(\{i\} \mid ([i-1]\cap B\setminus R)) = f(B\setminus R).$$
    We conclude that $f(R \mid (B\setminus R)) = f(B) - f(B \setminus R) \leq l(B) - l(B \setminus R) = l(R)$.
\end{proof}
Let $\theta > 0$ be a parameter that we will specify later.

\begin{mylemma}\label{lem: orienteering swap improvement}
    Given a current solution $B$, assume we have found, for some $b\in B$, a $b$-$b$ walk $W$ on $\Dcal(B)$ of length at most $L$ that satisfies $f_B(V(W)) >  L + \size^3\theta$.
    Then there exists a feasible exchange $(B \setminus R) \cup A$ in $\Ical_1 \cap \Ical_2$ such that 
    \begin{enumerate}
        \item $(B \setminus R) \cup A$ is a common basis,
        \item $f((B \setminus R) \cup A) > f(B) + \theta$, and 
        \item this exchange can be found in polynomial time. 
    \end{enumerate}
\end{mylemma}

\begin{proof}[Proof of \Cref{lem: orienteering swap improvement}]
     We construct a collection of simple directed cycles $\{C_1, \hdots, C_q\}$ which act as proxy for $W$ as follows. Order the vertices visited by $W$ by their first appearance in $W$ to get $V(W) = \{e_1, e_2, \hdots, e_p\} \subseteq E$ with $e_1 = b$. Append $b$ to the end of this list as well and denote it by $e_{p+1}$. For each $i \in [p]$ define a new walk $W_i$ to be the shortest directed path from $e_i$ to $e_{i+1}$. Such a path exists since $W$ reaches $e_{i+1}$ from $e_i$. Let $\bar{W}$ be the resulting walk defined by concatenating $W_1$ up to $W_p$. Notice that $\bar{W}$ satisfies the following properties:
     \begin{enumerate}
         \item $V(\bar{W}) \supseteq V(W)$, 
         \item the length of $\bar{W}$ is at most the length of $W$ because shortest paths (and hence shortest walks) are chosen, and 
         \item the number of arcs traversed by $\bar{W}$ is at most $\size^2$ since every shortest path can be taken to use at most $\size$ arcs and $p \leq \size$.
     \end{enumerate}
Now $\bar{W}$ is a $b$-$b$ walk which we can decompose into simple directed cycles by partitioning its arc occurrences. This gives us the collection $\{C_1, \hdots, C_q\}$ where $q \leq \size^2$.

     Subadditivity guarantees that $$f_B(V(W)) \leq f_B(V(\bar{W})) \leq \displaystyle\sum_{j =1}^q f_B(V(C_j)).$$
     Recall that the length of the walk $\bar{W}$ is the sum over the arc lengths traversed, which is equal to 
     \begin{equation*}
     \sum_{j =1}^q\sum_{(u, v) \in C_j} l(u, v) = \displaystyle\sum_{j =1}^q l(V(C_j)).
     \end{equation*} 
     Thus the following two inequalities hold
     \begin{align*}
        L + \size^3 \theta &< \displaystyle\sum_{j =1}^q f_B(V(C_j)), \text{ and} \\
        L &\geq \displaystyle\sum_{j =1}^q l(V(C_j))
    \end{align*}
    where the second inequality follows from the fact that the length of $\bar{W}$ is at most the length of $W$ which is assumed to be bounded by $L$.
    
    By averaging, there must exist a cycle in $\Dcal(B)$ which satisfies $f_B(V(C)) > l(V(C)) + \size\theta$. This cycle is not yet a feasible exchange, but by \Cref{lem: fractional cycle repair}, we can decompose this cycle into a collection of feasible exchanges $(B \setminus R_j) \cup A_j$ for $j \in [m]$ such that 
    \begin{align*}
        f_B(V(C)) &\leq \sum_{j=1}^m\lambda_j f_B(A_j), \text{ and} \\
        l(V(C))  &= \sum_{j=1}^m \lambda_j l(R_j).
    \end{align*}
    The first inequality follows from fractional subadditivity (\Cref{lem: frac additivity})  and the second equality follows from linearity of length. 
    It follows that
    \begin{equation*}
        \sum_{j=1}^m\lambda_j f_B(A_j) > \sum_{j=1}^m \lambda_j l(R_j) + \size \theta .
    \end{equation*}
    Thus since $|C| \leq \size$, we may find one of these feasible exchanges $(B \setminus R_j) \cup A_j$ satisfying $f_B(A_j) >  l(R_j) + \theta$. Finally, we show that this exchange is indeed improving for $B$. We have that 
    \begin{align*}
        f((B\setminus R_j) \cup A_j) - f(B)& = f(A_j \mid (B\setminus R_j)) - f(R_j \mid (B\setminus R_j)) \\
        & \geq  f(A_j \mid B) - l(R_j) \\
        & > \theta,
    \end{align*}
    where the first inequality follows from submodularity and \Cref{claim: modular lower support}. 
\end{proof}

Next, we define a variant of $l$ that uses only discrete integer values. The purpose of this
is twofold: it allows a reduction to a version of Submodular Orienteering where all arc
lengths are polynomially bounded integers, which together with the reduction in \Cref{sec:lower-bounds} implies equivalence (up to constants) of Submodular Orienteering with polynomially bounded lengths and Submodular Matroid Intersection Basis.
Furthermore,
it avoids unnecessary dependence on the values of $f$ in the number of operations of our algorithm.
To this end, let $l'(u, e) \in \mathbb Z$ with $l(u, e) \in [l'(u,e) \cdot \theta, (l'(u,e) + 1) \cdot \theta)$. Here $\theta$ is the unspecified parameter from earlier.
Note that $l'(u, e)$ is non-negative.
We now derive a similar lemma to \Cref{lem: orienteering swap improvement}, but with respect to $l'$.
\begin{mylemma}\label{lem: orienteering swap improvement discrete}
    Given a current solution $B$, assume we have found a $b$-$b$ walk $W$ on $\Dcal$
    satisfying $f_B(V(W)) > L' \theta + 2 \size^3 \theta$ of $l'$-length at most $L'$ for some $b\in B$. Then there exists a feasible exchange $(B\setminus R)\cup A$
    in $\Ical_1 \cap \Ical_2$ such that
     \begin{enumerate}
        \item $(B \setminus R) \cup A$ is a common basis,
        \item $f((B \setminus R) \cup A) > f(B) + \theta$, and 
        \item this exchange can be found in polynomial time. 
    \end{enumerate}
\end{mylemma}
\begin{proof}
    As in \Cref{lem: orienteering swap improvement}, we may assume without loss of generality that $W$ contains at most $\size^2$ arc occurrences.
    Recall that $l(u, e) \le (l'(u, e) + 1) \theta$ for each arc $(u, e)$.
    It follows that the $l$-length of $W$ denoted by $L$ satisfies
    \begin{equation*}
        L \le (L' + \size^2) \cdot \theta < f_B(V(W)) - \size^3 \theta .
    \end{equation*}
    The statement now follows from \Cref{lem: orienteering swap improvement}.
\end{proof}

Thus \Cref{lem: orienteering swap improvement discrete} states that we are able to find an improving swap for $B$ that maintains a common basis but increases the submodular value, if we have found a suitable budget $L'$ and walk on $\Dcal(B)$. Therefore, in one iteration of our algorithm, we are given a common basis $B$ and budget $L'$, and we find an improving feasible exchange when one exists for the given input. This is summarized in \Cref{alg:orienteer-to-cycle}.
Notice that a $b$-$b$ walk that is stationary does not have positive reward $f_B$.
So, if there exists an improving non-stationary walk, it will be returned by \Cref{alg:orienteer-to-cycle}.

\begin{algorithm}[H]
        \DontPrintSemicolon
        \SetKw{KwBy}{by}
        \caption{One iteration of Submodular Matroid Intersection Basis algorithm}
        \label{alg:orienteer-to-cycle}
        \KwIn{Common basis $B$ of $\Mcal_1$ and $\Mcal_2$, parameter $\theta >0$ quantifying the value improvement, budget $L'$.}
        \KwOut{Common basis $\bar{B}$ of $\Mcal_1$ and $\Mcal_2$ with $f(\bar{B}) \geq f(B)$.}
        Assign arc lengths $l'$ to $\Dcal(B)$ as described above. \;
        \For{every $b \in B$}{
        Use \Cref{lem: orienteering-output} to find a $b$-$b$ walk $W$ of $\Dcal(B)$ of $l'$-length at most $L'$ approximately maximizing $f_B(V(W))$.\;
        
        \If{$f_B(V(W)) > L'\theta + 2\size^3\theta$}{
        Apply \Cref{lem: orienteering swap improvement discrete} to obtain a feasible exchange $(B \setminus R) \cup A$ satisfying $f((B \setminus R) \cup A) > f(B) + \theta$\;
        \Return $(B \setminus R) \cup A$}
        
        }
         
        \Return $B$
\end{algorithm}

We now show how to apply \Cref{alg:orienteer-to-cycle} repeatedly over budget guesses to get an approximation algorithm for Submodular Matroid Intersection Basis. 

\subsection{The full maximum Submodular Matroid Intersection Basis algorithm}
Before we start the main algorithm, we assume that the initial input is a common basis $B_0$ such that $f(B_0) \geq \frac{f(\spm)}{K}$. Such a basis can be found in polynomial time by trying every element $e \in E$ and extending to a common basis by contraction and matroid intersection. The one with the largest submodular value must have value at least $\frac{f(\spm)}{K}$. This is because $f(\spm) \leq \sum_{e \in \spm} f(e) \leq K \max_{e \in E}f(e)$ where we have used \Cref{lem: frac additivity} in the first inequality. Note that we can assume there is a singleton of value $>0$ otherwise $f(\spm) = 0$ and the theorem is trivial.

The main algorithm is summarized in \Cref{alg: qptas}. It consists of guessing a budget and repeatedly applying \Cref{alg:orienteer-to-cycle} for enough iterations to obtain a solution that has submodular value at least $\frac{1}{\alpha + 1 + \epsilon} f(\spm)$. Here $\alpha$ is the approximation guarantee of finding an optimal $b$-$b$ walk. This $\alpha$ loss is inherent in the call to \Cref{lem: orienteering-output} by \Cref{alg:orienteer-to-cycle}. The budget $L'$ is the $l'$-length of a walk, and we would like to find a walk that is a simple cycle in $\Dcal(B)$. Thus the $l$-length of a walk is between $0$ and $l(B) = f(B)$, and the budget $L'$ is therefore guessed for all integers between $0$ and $\lfloor \sfrac{f(B)}{\Theta} \rfloor$.

\begin{algorithm}[H]
        \DontPrintSemicolon
        \SetKw{KwBy}{by}
        \caption{Maximum Submodular Matroid Intersection Basis algorithm}
        \label{alg: qptas}
        \KwIn{Common basis $B$ of $\Mcal_1$ and $\Mcal_2$, parameter $\epsilon >0$ quantifying the value improvement, approximation guarantee $\alpha$ of Submodular Orienteering.}
        \KwOut{Common basis $\tilde{B}$ of $\Mcal_1$ and $\Mcal_2$ with $f(\tilde{B}) \geq \frac{1}{\alpha + 1 + \epsilon}\cdot f(\spm)$.}

        Let $\improve \gets $True \hfill \tcp{Local optimality indicator.} 
        \While{$\improve$}{
        Let $ \theta \gets \frac{\epsilon}{2 \alpha K \size^3} f(B)$\hfill \tcp{Value of a single improvement.}

        $\improve \gets $False\;
        \For{every budget $L' \in \{0, 1, \dotsc, \lfloor f(B) / \theta \rfloor \}$}{$\bar{B} \gets$ output of \Cref{alg:orienteer-to-cycle} on $B, \Mcal_1, \Mcal_2, \theta, L'$\;
        \If{$f(\bar{B}) \geq f(B) + \theta$}{$\improve \gets $True\;  $B \gets \bar{B}$ \;\textbf{break}\hfill\tcp{Leave the for-loop.}}
        }
         }
        \Return $B$
\end{algorithm}

Recall our main theorem on the approximation guarantee of Submodular Matroid Intersection Basis:
\orienteerreduce*

We prove that \Cref{alg: qptas} outputs a solution satisfying the guarantees above. We will first show the approximation guarantee of \Cref{alg: qptas} and we will then analyze its runtime.  

\subsection{Analysis of approximation guarantee}

Let $B$ be the output of \Cref{alg: qptas}. Notice that termination of \Cref{alg: qptas} guarantees that in the final iteration, all choices of budgets $L'$ result in an output $\bar{B}$ of \Cref{alg:orienteer-to-cycle} with $f(\bar{B}) < f(B) + \theta$. Consider $\Dcal(B)$ and fix any simple directed cycle $C$ in $\Dcal(B)$. Such a cycle has length $l(V(C)) \le l(B) = f(B)$, so there exists some $L' \in \{0,1,\dotsc,\lfloor f(B) / \theta \rfloor\}$ satisfying $L' = \lfloor l(V(C)) / \theta \rfloor$.
Note that $l'(V(C)) \le L'$ since $l'(V(C)) \in \ZZ$.

If this cycle has value $f_B(V(C)) > \alpha L'\theta +  2 \alpha \size^3 \theta$, then there exists a $b$-$b$ walk $W^*$ on $\Dcal(B)$ satisfying $f_B(V(W^*)) > \alpha L'\theta + 2\alpha \size^3 \theta$ for some $b\in V(C)\cap B$. This walk contains at most $2K = 2|B|$ vertices, since it is alternating
between $B$ and $E\setminus B$.
Recall that we run Submodular Orienteering on every possible $b \in B$.
Once $b \in V(C)\cap B$ is guessed, the Submodular Orienteering algorithm outputs a walk $W$ such that its $l'$-length is at most $L'$ and $f_B(V(W)) \geq \frac{f_B(V(W^*))}{\alpha} > L' + 2\size^3\theta$. Then at least one of the calls to \Cref{alg:orienteer-to-cycle} returns an exchange $(B \setminus R) \cup A$ such that $f((B\setminus R) \cup A) > f(B) + \theta$ by \Cref{lem: orienteering swap improvement discrete}. But this contradicts the termination of \Cref{alg: qptas}. Thus every cycle in $\Dcal(B)$ must satisfy $f_B(V(C)) \leq \alpha \lfloor l(V(C)) / \theta \rfloor \theta + 2\alpha \size^3 \theta \le \alpha l(V(C)) + 2\alpha \size^3 \theta$.

We therefore take all cycles in $\Dcal(B)$ between $\spm$ and $B$ and sum this inequality to get a bound on $f(B)$ with respect to $f(\spm)$. By \Cref{lem: disj cycle decomposition}, there exists vertex disjoint cycles $C_1, \hdots, C_k$ in $\Dcal(B)$ such that $\spm= B \Delta \lb \bigcup_{i = 1}^k V(C_i)\rb$. Then $$f(\spm) = f\lb B \Delta \lb \bigcup_{i = 1}^k V(C_i)\rb\rb 
         \leq f\lb B \cup \bigcup_{i =1}^k V(C_i) \rb 
          = f_B\lb \bigcup_{i =1}^k V(C_i)\rb + f(B)$$
by monotonicity of $f$. Using submodularity and the fact that every cycle satisfies $f_B(V(C_i)) \leq \alpha l(V(C_i)) + 2\alpha \size^3 \theta$, we get that
\begin{align*}
    f(\spm)
        & \leq \displaystyle\sum_{i =1}^k f_B(V(C_i)) + f(B) \\
        & \leq \displaystyle\sum_{i =1}^k 
        \ld \alpha l(V(C_i)) + 2\alpha \size^3 \theta \rd + f(B) \\
        & \leq (1 + \alpha) f(B) + 2\alpha \size^3\theta k,
\end{align*}
where the second inequality follows from linearity of lengths and $l(B) = f(B)$. Notice that $k\leq K$, where $K$ is the size of a common base in $\Ical_1 \cap \Ical_2$. Then by the choice of $\theta$, $f(\spm) \leq (1 + \alpha + \epsilon) f(B)$ and we conclude the analysis of the approximation rate.

\subsection{Analysis of running time}
\Cref{alg:orienteer-to-cycle} makes at most $n$ many calls to Submodular Orienteering and runs in polynomial time. The decomposition into feasible exchanges also runs in time polynomial in $\size$.
Notice that by choice of $\theta$, the $l'$-lengths and $L'$ are always at most $2\alpha K \size^3 /\epsilon$ in any call to \Cref{alg:orienteer-to-cycle}, i.e., they are polynomially bounded integers. Then the full maximum common basis algorithm, \Cref{alg: qptas}, calls \Cref{alg:orienteer-to-cycle} whenever there still exists a $\theta$-improving exchange and for every budget in the budget guesses. Notice that by the discussion before, there are only polynomially many budget guesses. Moreover, every iteration guarantees $f(\bar{B}) \geq f(B) \lb 1 + \frac{\epsilon}{2\alpha K\size^3}\rb$ and we start with a basis $B_0$ satisfying $f(B_0) \geq \frac{f(\spm)}{K}$. Then after $w$ many iterations, $f(B) \geq \lb 1 + \frac{\epsilon}{2\alpha K\size^3}\rb^w f(B_0)$ so that after $\Theta\lb \frac{\alpha K \size^3}{\epsilon} \log K\rb$ many iterations, the value would exceed $f(\spm)$. Thus the algorithm must terminate before that.

\subsection{Exact Fair Matroid Monotone Submodular Maximization}\label{sec: exact fairness}

\label{sec:fairness}
Finally, we show how the above approximation algorithm for Submodular Matroid Intersection Basis (\Cref{alg: qptas}) can be used to achieve matching guarantees for Fair Matroid Monotone Submodular Maximization. The equivalence between the two problems was observed in \cite{halabiFairnessSubmodularMaximization2024}, and we include a different reduction here. In contrast to \cite{halabiFairnessSubmodularMaximization2024}, this reduction does not use any properties of $f$ such as monotonicity. Recall the following theorem with respect to solving this problem exactly. 

\exactfairness*


\begin{proof}
We show here that this can be reduced to a Submodular Matroid Intersection Basis problem.
The reduction is as follows. We may assume, for example, through enumeration, that we know
the cardinality $K$ of the optimal solution. Note that we have $\sum_{i=1}^r \ell_i \le K \le \sum_{i=1}^r k_i$.
We restrict $\mathcal I$ to sets of cardinality at most $K$. Note that this retains the matroid
structure.

We now construct a gammoid that models the fairness constraint. Recall that
a gammoid is a matroid defined from a directed graph $G = (V, A)$ with a set of possible sources $S \subseteq V$
and a set of possible destinations $N\subseteq V$.
A set $T\subseteq N$ is independent if there are vertex disjoint
paths from $S$ to $T$.
To model our fairness constraints, we construct a graph with three layers $V = V_1 \cup V_2 \cup V_3$.
$V_1$ consists of $K$ source vertices $S$. More precisely, $V_1 = S = L_1\cup\cdots\cup L_r \cup F$,
where $L_i$ are $\ell_i$ many vertices for each $i\in [r]$ and $F$ consists of $K - \sum_{i=1}^r \ell_i$ many
vertices. We set $V_2 = U_1 \cup \cdots \cup U_r$, where $U_i$ are $k_i$ many vertices for each $i\in [r]$.
There is an arc from each vertex in $L_i\cup F$ to each vertex in $U_i$ for $i\in[r]$.
Finally, $E = V_3$ consists of one vertex for each element of the ground set of the matroid $\Mcal$.
We add an arc from each vertex in $U_i$ to each vertex in $E_i$ for $i\in [r]$.
It is easy to see that the bases of this gammoid are exactly the sets
\begin{equation*}
        \mathcal F := \{ S\subseteq E : |S| = K \text{ and } \ \ell_i \le |S\cap E_i| \le k_i  \quad\forall i \in [r]\} \ .
\end{equation*}
It therefore suffices to find a common base of $(E,\mathcal I)$ and $(E,\mathcal F)$ that maximizes the submodular function. Note that the reduction is only for \textit{exact} common basis and not for a cardinality relaxation. Nevertheless, \Cref{alg: qptas} gives an algorithm which returns a common matroid basis whose guarantee therefore translates to exactly satisfying the lower and upper fairness constraints in Fair Matroid Monotone Submodular Maximization. The reduction concludes \Cref{coro:exact-lower-bounds}.
\end{proof}
\section{Hardness of a maximization over perfect matchings}\label{sec:lower-bounds}

In this section, we prove Theorem~\ref{thm:hardness}, which says that a better than logarithmic approximation is not possible for maximization of a monotone submodular function over perfect matchings in a bipartite graph.

Our lower bounds are based on hardness of the Group Steiner Tree problem,
specifically, a hardness result from~\cite{10.1137/20M1312988}, which itself is a refinement
of a result in~\cite{halperin2003polylogarithmic}. The hardness results of \cite{10.1137/20M1312988} assume that the projection games conjecture holds and that $\mathrm{NP}\not\subseteq \bigcap_{\epsilon > 0}\mathrm{ZPTIME}(2^{n^\epsilon})$. Weaker hardness guarantees that only use the assumption $\mathrm{NP}\not\subseteq \mathrm{ZPTIME}(n^{\mathrm{polylog}(n)})$ can also be obtained directly from \cite{halperin2003polylogarithmic}.

Group Steiner Tree is defined as follows:
We are given an undirected graph $G = (V,E)$ with edge weights $w\colon E \to \RR$. Furthermore, there
are groups $g_1,\dotsc,g_k\subseteq V$. A feasible solution is a connected subgraph that contains at least one
vertex from each group. Our goal is to find a minimum weight feasible solution.
We refer to the case where all weights are equal to $1$ as the \emph{unit weight case}.


\begin{mythm}[Theorem B.5 in~\cite{10.1137/20M1312988}]
	\label{prop:gst-hardness}
	There is no quasi-polynomial time $\omega(\log\log n / \log^{2} n)$-ap\-prox\-i\-mation
	algorithm for Group Steiner Tree assuming $\mathrm{NP}\not\subseteq \bigcap_{\epsilon > 0}\mathrm{ZPTIME}(2^{n^\epsilon})$ and the projection games conjecture holds, even on instances with $\log n = \Theta(\log k)$.\footnote{The restriction $\log k = \Theta(\log n)$ is not in the statement of Theorem C.7, but it is explicitly mentioned in the proof.}
\end{mythm}
In our reductions, we will use hardness for the unit weight case, which is not explicit in~\cite{halperin2003polylogarithmic, 10.1137/20M1312988}. However, it follows easily by standard rounding arguments.
\begin{mycor}
	There is no quasi-polynomial time $\omega(\log\log n / \log^{2} n)$-approximation
	algorithm for Group Steiner Tree assuming $\mathrm{NP}\not\subseteq \bigcap_{\epsilon > 0}\mathrm{ZPTIME}(2^{n^\epsilon})$ and the projection games conjecture holds, even on unit weight instances with $\log k = \Theta(\log n)$.
\end{mycor}
\begin{proof}
	Suppose that there is such an algorithm for the unit weight case.
	We extend it to the weighted case as follows: Guess the maximal weight $w_{\max}$ of an edge in the solution.
	We can remove every edge with weight greater than $w_{\max}$ from the instance. Furthermore, we contract each edge of weight less than
	$w_{\max} / |E|$. For any solution in the resulting instance, we can obtain a solution of
	the original instance by including all contracted edges whose endpoints are vertices of the solution.
	The cost of the solution will increase by at most $\frac{w_{\max}}{|E|} \cdot |E| \le \opt$.
	Hence, it suffices to design an approximation algorithm for the new instance.

	Note that all weights in the new instance are in the range $[\sfrac{w_{\max}}{|E|}, w_{\max}]$.
	We round each weight to the next higher integer multiple of $\sfrac{w_{\max}}{|E|}$. This transformation again
	loses a factor of at most $2$ in the approximation rate, since it can at most double each weight.
	Finally, we replace each edge of weight
	$i\cdot\frac{w_{\max}}{|E|}$ by a path of $i$ edges, each with weight $1$. Since $i\le |E|$, the 
	transformation increases the instance size only by a polynomial factor. 
    Furthermore, it is approximation-preserving. We solve it using the presumed algorithm for the unit weight case,
	which yields a $\omega(\log\log n / (4\log^2 n))$-approximation for Group Steiner Tree. 
	By \Cref{prop:gst-hardness}, this implies $\mathrm{NP}\subseteq \bigcap_{\epsilon > 0}\mathrm{ZPTIME}(2^{n^\epsilon})$ or that the projection games conjecture is false.
\end{proof}

In the following theorem, we prove hardness for Submodular Orienteering in the unbounded case, where we do not have a length constraint. The proof is an easy modification
of a proof by Chekuri and Pal~\cite{ChekuriPal05}, which was stated only for the variant with an arbitrary length bound.
The hardness for unbounded instances is important for our later reduction from matching.
We further restrict to the layered case, where the vertices are grouped in layers $V_1\sqcup\cdots\sqcup V_H$ with $s\in V_1$, $t\in V_H$ and all arcs going only from one layer to the next.
We note that unbounded instances are equivalent to instances with polynomially bounded integer lengths: For a given length bound $L$, we can change the
vertex set from $V$ to $V\times [L]$, where the second component encodes the length of the walk so far. The arc set and submodular function are translated in the natural way. We omit the details here for brevity.

\begin{mythm}\label{th:unbounded-lb}
	For any $\epsilon > 0$, there is no quasi-polynomial time $\omega(\log\log n / \log n)$-approximation algorithm for Submodular Orienteering
	assuming $\mathrm{NP}\not\subseteq \bigcap_{\epsilon > 0}\mathrm{ZPTIME}(2^{n^\epsilon})$ and the projection games conjecture holds.
    This holds even for the unbounded layered case.
\end{mythm}

\begin{proof}
	Suppose that there is a $\sfrac{1}{\alpha}$-approximation algorithm for Submodular Orienteering over graphs $G = (V, A)$ for $\alpha = o(\log n / \log\log n)$, where $n\coloneqq |V|$.
	Using it, we derive an approximation algorithm for Unit Weight Group Steiner Tree.
	Let $G = (V, E)$ and $g_1,\dotsc,g_k\subseteq V$ be the instance of Group Steiner Tree and let $T$ be the number of edges in the optimal solution.
    Further, assume that $\log |V| = \Theta(\log k)$.
	Our instance of Submodular Orienteering will seek a walk over $G$ that visits
	vertices from many of the groups. For technical reasons, let us assume that we know a vertex $r\in V$ in the solution to Group Steiner Tree. Both $T$ and $r$ can be obtained by a standard guessing framework. 

	Note that we can turn a solution for Group Steiner Tree with $T$ edges into a walk from $r$ to $r$ that
	visits all groups and contains exactly $2T$ edges: By duplicating each edge of the solution,
	we obtain a Eulerian graph, which can be turned into a closed walk.
	We construct the following instance of Submodular Orienteering to search for a closed walk that
	is approximately as good:
	There are layers $V_1,\dotsc,V_{2T+1}$, each of which is a copy of $V$.
	There is an arc from $u\in V_i$ to $v\in V_{i+1}$ if $\{u,v\}$ is an edge in the Group Steiner Tree
	instance. We define the submodular function $f\colon  V_1\sqcup\cdots\sqcup V_{2T+1} \to \RR$ with
	$f(U)$ being the number of groups intersecting $U$. We set $s$ to be the copy of $r$ in $V_1$
	and $t$ to be the copy of $r$ in $V_{2T+1}$.
	If the Group Steiner Tree instance has an optimum of $T$, then the optimum of the Submodular
	Orienteering instance is $k$. Thus, the presumed algorithm will compute a closed walk of 
	length $2T$ that visits at least $\sfrac{k}{\alpha}$ vertices.

	This walk is a connected subgraph that covers $\sfrac{k}{\alpha}$
	groups. We remove all groups that were covered and we repeat until no group is left.
	Finally, we output the union of the solutions from all iterations.
	The number of iterations is at most $O(\alpha \log k) = O(\alpha \log n)$ and each iteration has at most twice the  cost of the optimal Group Steiner Tree solution.
    We have thus obtained an $\sfrac{1}{\beta}$-approximation with $\beta = O(\alpha \log n) = o(\log^2 n / \log\log n)$, which by \Cref{prop:gst-hardness} implies
	$\mathrm{NP}\subseteq \bigcap_{\epsilon > 0}\mathrm{ZPTIME}(2^{n^\epsilon})$ or that the projection games conjecture is false.
\end{proof}

We now recall our main theorem which claims the same hardness for maximizing a submodular function over bipartite perfect matchings.

\hardness*

\begin{proof}
	Let $G = (V, A)$ be a unbounded layered instance of Submodular Orienteering with submodular function $f$ and start and terminal nodes $s, t$.
	In the following, we define an undirected bipartite graph $B = (V^+\cup V^-, E)$.
	Informally, $B$ is derived by replacing each vertex, except for source and sink, by a pair of vertices,
	where one of them is incident to all outgoing arcs and the other is incident to all incoming arcs of the original graph. Furthermore,
	each such pair of vertices is connected by a new edge. See \Cref{fig:reduction} for an illustration.
	Formally, let
	\begin{equation*}
		V^+ \coloneqq \{v^+ : v\in V\setminus \{t\}\},\quad V^- \coloneqq \{v^- : v\in V\setminus \{s\}\} \ .
	\end{equation*}
	Furthermore, we define 
	\begin{equation*}
		E \coloneqq \{\{v^-, v^+\} : v\in V\setminus \{s,t\}\} \cup \{\{u^+, v^-\} : (u,v)\in A\} \ .
	\end{equation*}
	We let $f' \colon 2^E\to \RR$ with
	\begin{equation*}
		f'(F) = f\left(\bigcup_{\substack{\{u^+, v^-\}\in F \\ u\neq v}} \{u,v\}\right) \qquad \forall F\subseteq E .
	\end{equation*}
	In other words, $f'$ is the value of all vertices contained in the given edges except for the edges that are between
	pairs of vertices created from the same vertex in $G$. If $f$ is monotone submodular then
	so is $f'$.

	For each perfect matching $M$ in $B$, there is an $s$-$t$ path $P$ in $G$ with $f(V(P)) = f'(M)$ and vice versa. Note that since $G$ is layered, every walk is a path.
    
	Consider a perfect matching $M$ in $B$.
	Let $P \coloneqq \{(u,v) \in A : \{u^+, v^-\} \in M\}$.
	It is easy to see that every vertex in $V\setminus\{s,t\}$ either has exactly one incoming and one outgoing arc in $P$ or has no incident arcs in $P$.
	Furthermore, $s$ has exactly one outgoing arc in $P$ and $t$ has exactly one incoming arc in $P$.
	Since $G$ is acyclic, $P$ must be the arcs of an $s$-$t$ path and $f(V(P)) = f'(M)$.

	Now consider an $s$-$t$ path $P$ in $G$. Let $M = \{\{u^+,v^-\} : (u,v)\in A(P)\} \cup \{\{v^-, v^+\} : v\in V\setminus V(P)\}$.
	Then $M$ is a perfect matching in $B$ and $f'(M) = f(V(P))$.

	Optimizing $f'$ over perfect matchings in $B$ is therefore equivalent to the instance of Submodular Orienteering, which together with \Cref{th:unbounded-lb} concludes the proof.
\end{proof}
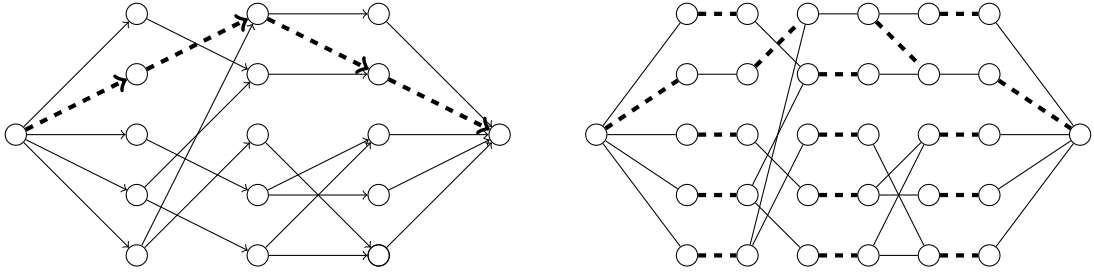
\begin{figure}
	\centering
		\begin{tikzpicture}[scale=0.8, every node/.style={inner sep=0pt,minimum size=8pt}]
		\node[draw,circle](s) at (0, 3) {};
		\node[draw,circle](v11) at (2, 1) {};
		\node[draw,circle](v12) at (2, 2) {};
		\node[draw,circle](v13) at (2, 3) {};
		\node[draw,circle](v14) at (2, 4) {};
		\node[draw,circle](v15) at (2, 5) {};
		\draw[->] (s) -- (v11);
		\draw[->] (s) -- (v12);
		\draw[->] (s) -- (v13);
		\draw[->, ultra thick, dashed] (s) -- (v14);
		\draw[->] (s) -- (v15);
		\node[draw,circle](v21) at (4, 1) {};
		\node[draw,circle](v22) at (4, 2) {};
		\node[draw,circle](v23) at (4, 3) {};
		\node[draw,circle](v24) at (4, 4) {};
		\node[draw,circle](v25) at (4, 5) {};
		\draw[->] (v11) -- (v23);
		\draw[->] (v11) -- (v25);
		\draw[->] (v12) -- (v21);
		\draw[->] (v12) -- (v24);
		\draw[->] (v13) -- (v22);
		\draw[->, ultra thick, dashed] (v14) -- (v25);
		\draw[->] (v15) -- (v24);
		\node[draw,circle](v31) at (6, 1) {};
		\node[draw,circle](v32) at (6, 2) {};
		\node[draw,circle](v33) at (6, 3) {};
		\node[draw,circle](v34) at (6, 4) {};
		\node[draw,circle](v35) at (6, 5) {};
		\draw[->] (v21) -- (v31);
		\draw[->] (v21) -- (v33);
		\draw[->] (v22) -- (v32);
		\draw[->] (v22) -- (v33);
		\draw[->] (v23) -- (v31);
		\draw[->] (v24) -- (v34);
		\draw[->] (v25) -- (v35);
		\draw[->, ultra thick, dashed] (v25) -- (v34);
		\node[draw,circle](v31) at (6, 1) {};
		\node[draw,circle](t) at (8, 3) {};
		\draw[->] (v31) -- (t);
		\draw[->] (v32) -- (t);
		\draw[->] (v33) -- (t);
		\draw[->, ultra thick, dashed] (v34) -- (t);
		\draw[->] (v35) -- (t);
	\end{tikzpicture}
	\qquad
		\begin{tikzpicture}[scale=0.8, every node/.style={inner sep=0pt,minimum size=8pt}]
		\node[draw,circle](s) at (0, 3) {};
		\node[draw,circle](u11) at (1.5, 1) {};
		\node[draw,circle](u12) at (1.5, 2) {};
		\node[draw,circle](u13) at (1.5, 3) {};
		\node[draw,circle](u14) at (1.5, 4) {};
		\node[draw,circle](u15) at (1.5, 5) {};
		\node[draw,circle](v11) at (2.5, 1) {};
		\node[draw,circle](v12) at (2.5, 2) {};
		\node[draw,circle](v13) at (2.5, 3) {};
		\node[draw,circle](v14) at (2.5, 4) {};
		\node[draw,circle](v15) at (2.5, 5) {};
		\draw[-] (s) -- (u11);
		\draw[-] (s) -- (u12);
		\draw[-] (s) -- (u13);
		\draw[-, ultra thick, dashed] (s) -- (u14);
		\draw[-] (s) -- (u15);

		\draw[-, ultra thick, dashed] (u11) -- (v11);
		\draw[-, ultra thick, dashed] (u12) -- (v12);
		\draw[-, ultra thick, dashed] (u13) -- (v13);
		\draw[-] (u14) -- (v14);
		\draw[-, ultra thick, dashed] (u15) -- (v15);
		\node[draw,circle](u21) at (3.5, 1) {};
		\node[draw,circle](u22) at (3.5, 2) {};
		\node[draw,circle](u23) at (3.5, 3) {};
		\node[draw,circle](u24) at (3.5, 4) {};
		\node[draw,circle](u25) at (3.5, 5) {};
		\node[draw,circle](v21) at (4.5, 1) {};
		\node[draw,circle](v22) at (4.5, 2) {};
		\node[draw,circle](v23) at (4.5, 3) {};
		\node[draw,circle](v24) at (4.5, 4) {};
		\node[draw,circle](v25) at (4.5, 5) {};
		\draw[-] (v11) -- (u23);
		\draw[-] (v11) -- (u25);
		\draw[-] (v12) -- (u21);
		\draw[-] (v12) -- (u24);
		\draw[-] (v13) -- (u22);
		\draw[-, ultra thick, dashed] (v14) -- (u25);
		\draw[-] (v15) -- (u24);

		\draw[-, ultra thick, dashed] (u21) -- (v21);
		\draw[-, ultra thick, dashed] (u22) -- (v22);
		\draw[-, ultra thick, dashed] (u23) -- (v23);
		\draw[-, ultra thick, dashed] (u24) -- (v24);
		\draw[-] (u25) -- (v25);
		\node[draw,circle](u31) at (5.5, 1) {};
		\node[draw,circle](u32) at (5.5, 2) {};
		\node[draw,circle](u33) at (5.5, 3) {};
		\node[draw,circle](u34) at (5.5, 4) {};
		\node[draw,circle](u35) at (5.5, 5) {};
		\node[draw,circle](v31) at (6.5, 1) {};
		\node[draw,circle](v32) at (6.5, 2) {};
		\node[draw,circle](v33) at (6.5, 3) {};
		\node[draw,circle](v34) at (6.5, 4) {};
		\node[draw,circle](v35) at (6.5, 5) {};
		\draw[-] (v21) -- (u31);
		\draw[-] (v21) -- (u33);
		\draw[-] (v22) -- (u32);
		\draw[-] (v22) -- (u33);
		\draw[-] (v23) -- (u31);
		\draw[-] (v24) -- (u34);
		\draw[-] (v25) -- (u35);
		\draw[-, ultra thick, dashed] (v25) -- (u34);

		\draw[-, ultra thick, dashed] (u31) -- (v31);
		\draw[-, ultra thick, dashed] (u32) -- (v32);
		\draw[-, ultra thick, dashed] (u33) -- (v33);
		\draw[-] (u34) -- (v34);
		\draw[-, ultra thick, dashed] (u35) -- (v35);
		\node[draw,circle](t) at (8, 3) {};
		\draw[-] (v31) -- (t);
		\draw[-] (v32) -- (t);
		\draw[-] (v33) -- (t);
		\draw[-, ultra thick, dashed] (v34) -- (t);
		\draw[-] (v35) -- (t);
	\end{tikzpicture}
	\caption{Reduction from Orienteering to Matching}
	\label{fig:reduction}
\end{figure}

\section{Tool and Computational Resource Disclosure} The authors used GPT to assist with generating a simplified analysis of the quasi-polynomial approximation algorithm in \Cref{sec: submod common basis}. In particular, the tool was largely responsible for coming up with \Cref{lem: fractional cycle repair}, which simplified the rest of the proof. This was independently verified and extended by the authors. All writing was done by the authors.

\newpage
\bibliographystyle{alpha}
\bibliography{biblio}

@article{calinescu2011maximizing,
  title={Maximizing a monotone submodular function subject to a matroid constraint},
  author={Calinescu, Gruia and Chekuri, Chandra and P\'{a}l, Martin and Vondr{\'a}k, Jan},
  journal={SIAM Journal on Computing},
  volume={40},
  number={6},
  pages={1740--1766},
  year={2011},
  publisher={SIAM}
}

@inproceedings{filmus2012tight,
  title={A tight combinatorial algorithm for submodular maximization subject to a matroid constraint},
  author={Filmus, Yuval and Ward, Justin},
  booktitle={2012 IEEE 53rd Annual Symposium on Foundations of Computer Science},
  pages={659--668},
  year={2012},
  organization={IEEE}
}

@article{lee2010submodular,
  title={Submodular maximization over multiple matroids via generalized exchange properties},
  author={Lee, Jon and Sviridenko, Maxim and Vondr{\'a}k, Jan},
  journal={Mathematics of Operations Research},
  volume={35},
  number={4},
  pages={795--806},
  year={2010},
  publisher={INFORMS}
}

@article{mahabadi2026improved,
  title={Improved Algorithms for Fair Matroid Submodular Maximization},
  author={Mahabadi, Sepideh and Sarkar, Sherry and Tarnawski, Jakub},
  journal={Advances in Neural Information Processing Systems},
  volume={38},
  pages={152592--152612},
  year={2026}
}

@article{feige1998threshold,
  title={A threshold of ln n for approximating set cover},
  author={Feige, Uriel},
  journal={Journal of the ACM (JACM)},
  volume={45},
  number={4},
  pages={634--652},
  year={1998},
  publisher={ACM New York, NY, USA}
}

@article{celis2017multiwinner,
  title={Multiwinner voting with fairness constraints},
  author={Celis, L Elisa and Huang, Lingxiao and Vishnoi, Nisheeth K},
  journal={arXiv preprint arXiv:1710.10057},
  year={2017}
}

@inproceedings{chekuri2009dependent,
  title={Dependent randomized rounding via exchange properties of combinatorial structures},
  author={Chekuri, Chandra and Vondr{\'a}k, Jan and Zenklusen, Rico},
  booktitle={2010 IEEE 51st Annual Symposium on Foundations of Computer Science},
  pages={575--584},
  year={2010},
  organization={IEEE}
}

@inproceedings{chekuri2011multi,
  title={Multi-budgeted matchings and matroid intersection via dependent rounding},
  author={Chekuri, Chandra and Vondr{\'a}k, Jan and Zenklusen, Rico},
  booktitle={Proceedings of the twenty-second annual ACM-SIAM symposium on Discrete Algorithms},
  pages={1080--1097},
  year={2011},
  organization={SIAM}
}

@article{grandoni2014new,
  title={New approaches to multi-objective optimization},
  author={Grandoni, Fabrizio and Ravi, Ramamoorthi and Singh, Mohit and Zenklusen, Rico},
  journal={Mathematical Programming},
  volume={146},
  number={1},
  pages={525--554},
  year={2014},
  publisher={Springer}
}

@article{el2020fairness,
  title={Fairness in streaming submodular maximization: Algorithms and hardness},
  author={El Halabi, Marwa and Mitrovi{\'c}, Slobodan and Norouzi-Fard, Ashkan and Tardos, Jakab and Tarnawski, Jakub M},
  journal={Advances in Neural Information Processing Systems},
  volume={33},
  pages={13609--13622},
  year={2020}
}

@inproceedings{halabiFairnessSubmodularMaximization2024,
  title = {Fairness in {{Submodular Maximization}} over a {{Matroid Constraint}}},
  booktitle = {Proceedings of {{The}} 27th {{International Conference}} on {{Artificial Intelligence}} and {{Statistics}} (AISTAT)},
  author = {Halabi, Marwa El and Tarnawski, Jakub and {Norouzi-Fard}, Ashkan and Vuong, Thuy-Duong},
  year = 2024,
  pages = {1027--1035},
  publisher = {PMLR},
  issn = {2640-3498},
}

@article{vondrak2013symmetry,
  title={Symmetry and approximability of submodular maximization problems},
  author={Vondr{\'a}k, Jan},
  journal={SIAM Journal on Computing},
  volume={42},
  number={1},
  pages={265--304},
  year={2013},
  publisher={SIAM}
}

@article{VondrakCZ11,
author = {Chekuri, Chandra and Vondr{\'a}k, Jan and Zenklusen, Rico},
title = {Submodular Function Maximization via the Multilinear Relaxation and Contention Resolution Schemes},
journal = {SIAM Journal on Computing},
volume = {43},
number = {6},
pages = {1831-1879},
year = {2014},
doi = {10.1137/110839655}
}

@inproceedings{feldmanUnifiedContinuousGreedy2011,
  title = {A {{Unified Continuous Greedy Algorithm}} for {{Submodular Maximization}}},
  booktitle = {Proceedings of the 52nd {{Annual IEEE Symposium}} on {{Foundations}} of {{Computer Science}} ({{FOCS}})},
  author = {Feldman, M. and Naor, J. and Schwartz, R.},
  year = 2011,
  pages = {570--579},
  doi = {10.1109/FOCS.2011.46},
  isbn = {978-0-7695-4571-4},
}

@inproceedings{FeldmanNS11,
  author       = {Moran Feldman and
                  Joseph Naor and
                  Roy Schwartz},
  editor       = {Luca Aceto and
                  Monika Henzinger and
                  Jir{\'{\i}} Sgall},
  title        = {Nonmonotone Submodular Maximization via a Structural Continuous Greedy
                  Algorithm - (Extended Abstract)},
  booktitle    = {Automata, Languages and Programming - 38th International Colloquium,
                  {ICALP} 2011, Zurich, Switzerland, July 4-8, 2011, Proceedings, Part
                  {I}},
  series       = {Lecture Notes in Computer Science},
  volume       = {6755},
  pages        = {342--353},
  publisher    = {Springer},
  year         = {2011}
}

@article{madiman2010information,
  title={Information inequalities for joint distributions, with interpretations and applications},
  author={Madiman, Mokshay and Tetali, Prasad},
  journal={IEEE Transactions on Information Theory},
  volume={56},
  number={6},
  pages={2699--2713},
  year={2010},
  publisher={IEEE}
}

@inproceedings{ChekuriPal05,
  author       = {Chandra Chekuri and
                  Martin P{\'{a}}l},
  title        = {A Recursive Greedy Algorithm for Walks in Directed Graphs},
  booktitle    = {46th Annual {IEEE} Symposium on Foundations of Computer Science, {FOCS}
                  2005, Pittsburgh, PA, USA, October 23-25, 2005, Proceedings},
  pages        = {245--253},
  publisher    = {{IEEE} Computer Society},
  year         = {2005},
  url          = {https://doi.org/10.1109/SFCS.2005.9},
  doi          = {10.1109/SFCS.2005.9},
  bibsource    = {dblp computer science bibliography, https://dblp.org}
}

@article{brualdi1969comments,
  title={Comments on bases in dependence structures},
  author={Brualdi, Richard A},
  journal={Bulletin of the Australian Mathematical Society},
  volume={1},
  number={2},
  pages={161--167},
  year={1969},
  publisher={Cambridge University Press}
}

@inproceedings{10.1137/20M1312988,
  title={${O}(\log^2{k}/\log\log{k})$--approximation algorithm for directed steiner tree: a tight quasi-polynomial-time algorithm},
  author={Grandoni, Fabrizio and Laekhanukit, Bundit and Li, Shi},
  booktitle={Proceedings of the 51st Annual ACM SIGACT Symposium on Theory of Computing},
  pages={253--264},
  year={2019}
}

@inproceedings{halperin2003polylogarithmic,
  title = {Polylogarithmic {{Inapproximability}}},
  author={Halperin, Eran and Krauthgamer, Robert},
  booktitle = {Proceedings of {STOC}},
  pages={585--594},
  year={2003}  
}

@book{schrijver2003combinatorial,
  title={Combinatorial optimization: polyhedra and efficiency},
  author={Schrijver, Alexander},
  volume={24},
  year={2003},
  publisher={Springer Science \& Business Media}
}

@inproceedings{ene2016constrained,
  title={Constrained submodular maximization: Beyond 1/e},
  author={Ene, Alina and Nguyen, Huy L},
  booktitle={2016 IEEE 57th Annual Symposium on Foundations of Computer Science (FOCS)},
  pages={248--257},
  year={2016},
  organization={IEEE}
}

@inproceedings{buchbinder2024constrained,
  title={Constrained submodular maximization via new bounds for dr-submodular functions},
  author={Buchbinder, Niv and Feldman, Moran},
  booktitle={Proceedings of the 56th annual ACM symposium on theory of computing},
  pages={1820--1831},
  year={2024}
}

@article{buchbinder2019constrained,
  title={Constrained submodular maximization via a nonsymmetric technique},
  author={Buchbinder, Niv and Feldman, Moran},
  journal={Mathematics of Operations Research},
  volume={44},
  number={3},
  pages={988--1005},
  year={2019},
  publisher={INFORMS}
}

@incollection{buchbinder2018submodular,
  author    = {Niv Buchbinder and Moran Feldman},
  title     = {Submodular Functions Maximization Problems},
  booktitle = {Handbook of Approximation Algorithms and Metaheuristics, Volume 1: Methodologies and Traditional Applications},
  editor    = {Teofilo F. Gonzalez},
  edition   = {2nd},
  pages     = {753--788},
  year      = {2018},
  publisher = {Chapman and Hall/CRC},
  doi       = {10.1201/9781351236423},
  note      = {Also available as arXiv preprint arXiv:1802-07098}
}

@article{bilmes2022submodularity,
  author        = {Jeff Bilmes},
  title         = {Submodularity In Machine Learning and Artificial Intelligence},
  journal       = {arXiv preprint arXiv:2202.00132},
  year          = {2022},
  eprint        = {2202.00132},
  archivePrefix = {arXiv},
  primaryClass  = {cs.LG},
  url           = {https://arxiv.org}
}

@incollection{krause2014submodular,
  author    = {Andreas Krause and Daniel Golovin},
  title     = {Submodular Function Maximization},
  booktitle = {Tractability: Practical Approaches to Hard Problems},
  editor    = {Lucas Bordeaux and Youssef Hamadi and Pushmeet Kohli},
  pages     = {71--104},
  year      = {2014},
  publisher = {Cambridge University Press},
  doi       = {10.1017/CBO9781139177801.004},
  url       = {https://las.inf.ethz.ch/submodularity/}
}

@article{dutta2026submodular,
  author        = {Dutta, Shamak and Gharesifard, Bahman and Smith, Stephen L.},
  title         = {Submodular Optimization with Applications to Decision and Control},
  journal       = {arXiv preprint arXiv:2606.10192},
  year          = {2026},
  archivePrefix = {arXiv},
  eprint        = {2606.10192},
  primaryClass  = {math.OC},
  url           = {https://arxiv.org/abs/2606.10192}
}

@article{nemhauser1978analysis,
  author  = {Nemhauser, George L. and Wolsey, Laurence A. and
             Fisher, Marshall L.},
  title   = {An Analysis of Approximations for Maximizing Submodular Set
             Functions---{I}},
  journal = {Mathematical Programming},
  volume  = {14},
  number  = {1},
  pages   = {265--294},
  year    = {1978},
  doi     = {10.1007/BF01588971}
}

@article{nemhauser1978best,
  author  = {Nemhauser, George L. and Wolsey, Laurence A.},
  title   = {Best Algorithms for Approximating the Maximum of a Submodular
             Set Function},
  journal = {Mathematics of Operations Research},
  volume  = {3},
  number  = {3},
  pages   = {177--188},
  year    = {1978},
  doi     = {10.1287/moor.3.3.177}
}

@article{bach2013learning,
  author  = {Bach, Francis},
  title   = {Learning with Submodular Functions: A Convex Optimization
             Perspective},
  journal = {Foundations and Trends in Machine Learning},
  volume  = {6},
  number  = {2-3},
  pages   = {145--373},
  year    = {2013},
  doi     = {10.1561/2200000039}
}

@inproceedings{moshkovitz2012projection,
  title={The projection games conjecture and the NP-hardness of ln n-approximating set-cover},
  author={Moshkovitz, Dana},
  booktitle={International Workshop on Approximation Algorithms for Combinatorial Optimization},
  pages={276--287},
  year={2012},
  organization={Springer}
}

@article{rohwedder2026markov,
  author        = {Rohwedder, Lars and Zenklusen, Rico.},
  title         = {Strong and Compact Policies for Submodular Markov Decision Processes via LP-Based Submodular Orienteering},
  journal       = {arXiv preprint arXiv:2609.15539},
  year          = {2026},
  archivePrefix = {arXiv},
  eprint        = {2609.15539},
  primaryClass  = {cs.DS},
  url           = {https://arxiv.org/pdf/2609.15539}
}

\end{document}